\documentclass[11pt]{extarticle}

\usepackage{soul}
\usepackage{url}
\usepackage[hidelinks,breaklinks]{hyperref}
\usepackage[utf8]{inputenc}
\usepackage[small]{caption}
\usepackage{graphicx}
\usepackage{booktabs}

\usepackage{amsmath,multirow}

\usepackage{amsthm,thm-restate,thmtools}
\usepackage{amssymb}

\usepackage{enumerate}
\usepackage{etoolbox}
\usepackage{subcaption}
\usepackage{xcolor,xfrac}
\usepackage{natbib}
\usepackage{hyperref}
\usepackage[english]{babel}
\usepackage{graphicx,charter}
\usepackage{authblk}
\usepackage{framed}
\usepackage[normalem]{ulem}
\usepackage{amsmath,multirow}
\usepackage{amsthm,thm-restate,thmtools}
\usepackage{tikz}
\usepackage{amssymb}
\usepackage{amsfonts}
\usepackage{enumerate}
\usepackage[T1]{fontenc}
\usepackage[utf8]{inputenc}
\usepackage[top=1 in,bottom=1in, left=1 in, right=1 in]{geometry}
\usepackage{xcolor,xfrac}
\usepackage{url}
\usepackage{enumitem,natbib}
\usepackage{tabularx}
\usepackage{algorithm}
\usepackage{algorithmic}

\usepackage{xcolor}
\usepackage{thm-restate}
\usepackage{hyperref}
\hypersetup{
     colorlinks   = true,
     linkcolor    = red, 
     urlcolor     = teal, 
	 citecolor    = blue 
}
\usepackage{mathtools}

\newtheorem{example}{Example}
\newtheorem{lemma}{Lemma}
\newtheorem{proposition}{Proposition}
\newtheorem{corollary}{Corollary}
\newtheorem{claim}{Claim}
\newtheorem{remark}{Remark}

\usepackage[capitalise,noabbrev]{cleveref}
\Crefname{remark}{Remark}{Remarks}
\Crefname{rmk}{Remark}{Remarks}
\Crefname{dfn}{Definition}{Definitions}
\Crefname{thm}{Theorem}{Theorems}
\Crefname{cor}{Corollary}{Corollaries}
\Crefname{lem}{Lemma}{Lemmas}
\Crefname{examplex}{Example}{Examples}
\Crefname{prop}{Proposition}{Propositions}

\newcommand{\MMS}{\textrm{\textup{MMS}}}

\newcommand{\NPC}{\textrm{\textup{NP-complete}}}

\newcommand{\NPH}{\textrm{\textup{NP-hard}}}
\newcommand{\EFX}{\textrm{\textup{EFX}}}
\newcommand{\EF}[1]{\ifstrempty{#1}{\textrm{\textup{EF}}}{\textrm{\textup{EF{$#1$}}}}}
\newcommand{\RM}{\textrm{\textup{RM}}}
\newcommand{\PO}{\textup{PO}}

\newcommand{\pos}{\mathrm{pos}}

\title{Fairly Allocating Goods and (Terrible) Chores}

\author{Hadi Hosseini}
\author{Aghaheybat Mammadov}
\author{Tomasz Wąs}
\affil{Pennsylvania State University\vspace{0.15cm}\\
 {\normalsize \{hadi, mammadovagha, twas\}@psu.edu}}
\date{}

\begin{document}
\maketitle
\begin{abstract}
\noindent
We study the fair allocation of mixtures of indivisible goods and chores under lexicographic preferences---a subdomain of additive preferences. 
A prominent fairness notion for allocating indivisible items is \textit{envy-freeness up to any item} (\EFX{}). Yet, its existence and computation has remained a notable open problem.
By identifying a class of instances with ``terrible chores'', we  show that determining the existence of an \EFX{} allocation is \NPC{}. This result immediately implies the intractability of \EFX{} under additive preferences.
Nonetheless, we propose a natural subclass of lexicographic preferences for which an \EFX{} and Pareto optimal (\PO{}) allocation is guaranteed to exist and can be computed efficiently for any mixed instance.
Focusing on two weaker fairness notions, we investigate finding \EF{1} and \PO{} allocations for special instances with terrible chores, and show that \MMS{} and \PO{} allocations can be computed efficiently for any mixed instance with lexicographic preferences. 
\end{abstract}

\section{Introduction}

Fair division of indivisible items has provided a rich mathematical framework for studying computational and axiomatic aspects of fairness in a variety of settings ranging from assigning students to courses \citep{budish2011combinatorial} and distributing food donations \citep{aleksandrov2015online} to assigning papers to reviewers \citep{shah2022challenges,payan2022will} and distributing medical equipment and vaccines \citep{schmidt2021equitable,aziz2021efficient,pathak2021fair}.
In these applications, the preferences of agents over items may be \textit{subjective}, that is, some agents may consider an item as a \textit{good} (with non-negative utility) while others may see the same item as a \textit{chore} (with negative utility). 
For instance, in peer reviewing, reviewers may consider a paper to be a chore if it is outside of their immediate expertise while another subset of reviewers consider it as a good due its proximity to their own field. 
Thus, an emerging line of work has focused on fair allocation of mixture of goods and chores  \citep{ACI+19fair,BBB+20envy,kulkarni2021indivisible}.

When distributing indivisible items, a prominent fairness notion, \textit{envy-freeness} (\EF{}) \citep{foley1967resource,gamow1958puzzle}, may not always exist.
Its most compelling relaxation,
\textit{envy-freeness up to any item} (\EFX{}) \citep{CKM+19unreasonable},
states that any pairwise envy
is eliminated if we remove \emph{any} single item that is
considered a good in the envied agent's bundle
or is seen as a chore in the envious agent's bundle.
A slightly weaker notion is \textit{envy-freeness up to one item} (\EF{1}) \citep{LMM+04approximately,budish2011combinatorial}, which requires that any pairwise envy can be eliminated by the removal of \emph{some} single item from the bundle of one of the two agents.
%
These relaxations gave rise to several challenging open problems, particularly when dealing with chores: the existence of \EFX{} and the existence and computation of \EF{1} in conjunction with  efficiency notions such as \textit{Pareto optimality} (\PO{}).


To gain insights into structural and computational boundaries of achieving these fairness notions, several recent efforts have considered a variety of \textit{restricted domains} such as limiting the number of agents \citep{CGM20efx,mahara2021extension}, the item types \citep{aziz2022fairtypes,nguyen2023fair}, or the valuations (binary, bi-valued valuations, or identical) \citep{HPP+20fair,garg2022fair,BBB+20envy}.
One such natural restriction are \textit{lexicographic preferences}---a subdomain of additive preferences---which provides a compact representation of preferences, and has been studied in voting \citep{LMX18voting}, object allocation \citep{saban2014note,HL19multiple}, and fair division \citep{nguyen2020fairly,ebadian2022fairly}.


In this domain, it was recently shown that an \EFX{} allocation may not always exist for mixed instances \citep{hosseini2022fairly}. 
Moreover, while weaker fairness notions such as \EF{1} and maximin share (\MMS{}) are guaranteed to exist for mixed items, their computation along with \PO{} remains unknown even for \textit{objective} instances where all agents agree on whether an item is a good or a chore.

The non-existence of \EFX{} for mixed items crucially relies on a set of highly undesirable chores (aka `terrible chores'). Without these chores (i.e., if a single agent considers a good as its most important item), under lexicographic preferences an \EFX{} and \PO{} allocation can be computed in polynomial time \citep{hosseini2022fairly}.
This observation raises several important questions: Can we efficiently decide whether an \EFX{} allocation exists even in the presence of terrible chores?
Can we efficiently compute an \EF{1} (or \MMS{}) allocation in conjunction with Pareto optimality?

\subsection{Contributions}

We focus on the allocation of mixtures of goods and chores in the lexicographic domain and resolve several open computational problems pertaining to the well-studied fairness notions of \EFX{}, \EF{1}, and \MMS{}.

\paragraph[EFX.]{\EFX{}.}
We show that determining the existence of an \EFX{} allocation is \NPC{} under lexicographic mixed instances even for \textit{objective} preferences, i.e., when all agents agree on whether an item is a good or a chore   (\cref{thrm:efx:hardness}). To the best of our knowledge, this finding is the first computational intractability result for \EFX{} over any preference extension containing lexicographic (and as a result additive) preferences.
Subsequently, we discuss that deciding whether an \EFX{} and PO allocation exist is \NPH{} (\cref{cor:efx+po:hardness}).
%

\paragraph[EFX+PO.]{\EFX{}+\PO{}.}
Given the non-existence of \EFX{} even for objective mixed instances, and the computational hardness of determining such allocations, we identify a natural variation of lexicographic preferences, called \textit{separable} lexicographic preferences for which positive results can be obtained. In particular, we show that \EFX{} and \PO{} allocations are guaranteed to exist even on instances that contain terrible chores (\cref{thrm:efx+po+sep}), and thus, prove that under separable lexicographic preferences, an \EFX{} and \PO{} allocation can be computed efficiently (\cref{cor:efx+po}). 

\paragraph[EF1+PO.]{EF$\mathbf{1}$+\PO{}.}
Given the non-existence of \EFX{} under general (not necessarily separable) lexicographic mixed instances, we focus our attention on \EF{1} along with Pareto optimality. While an \EF{1} allocation always exists and can be computed efficiently \citep{BSV21approximate,ACI+19fair}, its existence and computation along with \PO{} remains open even for additive chores-only instances.
We identify a class of lexicographic mixed instances with sufficiently many common terrible chores for which an \EF{1} and \PO{} allocation can be computed in polynomial time (\cref{thm:ef1POterrible}), and discuss several technical challenges in extending these results. 


\paragraph[MMS+PO.]{\MMS{}+\PO{}.}
Despite the non-existence of \EFX{} and challenges in achieving \EF{1}+\PO{}, we show that an \MMS{} and \PO{} allocation always exist for any mixed instance containing terrible chores (\cref{thrm:mms+po}), and can be computed efficiently for any mixed instance (\cref{cor:mms+po}).
Moreover, we show that when the efficiency is strengthened to rank-maximality (RM), deciding whether an instance admits an \MMS{} and rank-maximal allocation is \NPC{} (\cref{thrm:mms+rm}).


\subsection{Related Work}

The existence of \EFX{} is a major open problem in goods-only and chores-only settings. Moreover, \EFX{} is known to be incompatible with \PO{} under non-negative valuations \citep{PR20almost}.
An \EFX{} allocation may fail to exist under non-monotone, non-additive, and identical valuation functions \citep{BBB+20envy} and for mixed items with additive valuations \citep{hosseini2022fairly}.
Yet, determining whether an instance admits an \EFX{} allocation has been an open question, which we answer in this paper.

An \EF{1} allocation can be computed efficiently in goods-only \citep{CKM+19unreasonable,LMM+04approximately} and chores-only \citep{ACI+19fair,BSV21approximate} settings. When considering economic efficiency, for goods-only problems \EF{1} is compatible with \PO{}~\citep{CKM+19unreasonable} and can be computed in pseudo-polynomial time~\citep{BKV18finding}. 
In contrast, for chores-only settings, it is not known whether \EF{1} and \PO{} allocations exist under additive valuations.
For mixed items, an \EF{1} allocation can still be computed efficiently when valuations are doubly monotonic (which includes additive valuations) \citep{BSV21approximate,ACI+19fair} through a careful use of the envy-graph algorithm.
However, achieving \EF{1} alongside \PO{} (except for two agents \citep{ACI+19fair}) remains an open problem.

With additive valuations, an \MMS{} allocation could fail to exist in both the goods-only~\citep{KPW18fair} and the chores-only~\citep{aziz2017algorithms} settings. Due to this non-existence, several multiplicative \citep{aziz2017algorithms,ghodsi2018fair,garg2020improved}  and ordinal approximations \citep{babaioff2019fair,hosseini2022ordinal} to \MMS{} have been proposed for both goods-only and chores-only settings. For mixed items, no constant multiplicative \citep{kulkarni2021indivisible} or ordinal 
 \citep{hosseini2022ordinalchores} approximation of \MMS{} may exist.

\paragraph{Domain Restriction.}
To circumvent the negative results and explore the computational boundary and their compatibility with other properties, much attention has been given to studying fairness in restricted domains.
For goods-only settings, an \EFX{} allocation is guaranteed to exist when agents have identical monotone valuations~\citep{PR20almost}, or submodular valuations with binary marginals~\citep{BEF21fair,VZ22yankee}, or additive valuations with at most two distinct values~\citep{ABF+21maximum,GM21computing}.
For for chores-only instances, an \EFX{} allocation can be efficiently computed  when there are four agents with only two types of additive valuations over seven items \citep{BT2022EFX}.
Under lexicographic preferences, \EFX{} and \PO{} allocation always exist and can be computed in polynomial time for goods-only and chores-only settings~\citep{hosseini2021fair}, and can often be satisfied along with strategyproofness and other desirable properties.

In chores-only settings, \EF{1} and \PO{} allocations can be computed in polynomial time when preferences are restricted to bivalued additive valuations \citep{ebadian2022fairly,garg2022fair} or when there are only two types of chores \citep{aziz2022fairtypes}.
Similarly, \MMS{} allocations are known to always exist for restricted domains such as personalized bivalued valuations, and can be computed efficiently along \PO{} under factored bivalued valuations and weakly lexicographic valuations (allowing ties between items) \citep{ebadian2022fairly}.


\section{Preliminaries}
For every $k \in \mathbb{N}$, let $[k]=\{1,\dots,k\}$.
Let $N \coloneqq [n]$ be a set of $n$ \emph{agents} and $M \coloneqq \{o_1, \ldots, o_m\}$ be a set of $m$ \emph{items}.
For each $i \in N$, $G_i \subseteq M$ denotes the subset of items
considered as \emph{goods} and $C_i \coloneqq M \setminus G_i$ is the set of items considered as \emph{chores} by agent $i$.
Items that are goods (chores) for all agents are referred to as \emph{common goods} (similarly, \emph{common chores}) i.e., $\bar{G} \coloneqq \bigcap_{i \in N} G_i$
($\bar{C} \coloneqq \bigcap_{i \in N} C_i$). 


%


\paragraph{Preferences.}
We consider \emph{lexicographic preferences} over all possible subsets of mixed items, through a linear order that specifies the \textit{importance ordering} of items for each agent.
Thus, for each agent $i \in N$ there is an associated importance ordering $\rhd_i$ that is a linear order over $M$. Thus, an \textit{importance profile} is simply denoted by $\rhd \coloneqq (\rhd_1, \ldots, \rhd_n)$.
We use ``$+$'' (or ``$-$'') in superscript to denote that an item is a good (or a chore) in an importance ordering. 
For example, 
\begin{equation}
\label{eq:example}
\rhd_i : \ o^+_1 \rhd o^-_2 \rhd o^+_3,
\end{equation}
means that agent $i$ considers $o_1$ and $o_3$ as goods and $o_2$ as a chore.
%
%
%
%
%
%
Given an importance ordering $\rhd_i$ and a subset of items $X \subseteq M$, let $\rhd_i(k, X)$ denote the $k$-th most important item 
in $X$ according to $\rhd_i$.
In case $X=M$, we will write simply $\rhd_i(k)$ for brevity.

Agent $i$'s lexicographic preference $\succ_i$ is a strict linear order over all possible subsets of items, which is defined based on its importance ordering $\rhd_i$ as follows:  
For every non-identical $X, Y \subseteq M$, we say that $X \succ_i Y$
if and only if
\(
    \rhd_i(1, X \triangle Y) \in (X \cap G_i) \cup (Y \cap C_i),
\)
where $\triangle$ denotes a symmetric difference
(formally, $A \triangle B = (A \cup B) \setminus (A \cap B)$
for every sets $A$ and $B$).
In other words, $X$ is preferred to $Y$ if and only if the most important item
on which they differ is either a good in $X$ or a chore in $Y$.
For every $X, Y \subseteq M$, we will write $X \succeq_i Y$ if $X \succ_i Y$ or $X = Y$. 
For instance, based on agent $i$'s preference stated above 
in~\eqref{eq:example},
we have
\(
    \{ o^+_1, o^+_3 \} \succ_i
    \{ o^+_1 \} \succ_i 
    \{ o^+_1, o^-_2, o^+_3 \} \succ_i 
    \{ o^+_1, o^-_2 \} \succ_i 
    \{ o^+_3 \} \succ_i 
    \emptyset \succ_i 
    \{ o^-_2, o^+_3 \} \succ_i 
    \{ o^-_2 \}.
\)
%
%
%

\paragraph{Terrible Chores and Separable Preferences.}
For disjoint subsets $X, Y \subseteq M$,
we use a shorthand notation $X \rhd_i Y$ to say that
for every $x \in X$ and $y \in Y$ it holds that
$x \rhd_i y$.
A set of \emph{terrible chores}
is a set of chores more important than any good, i.e.,
a maximal set $C^*_i \subseteq C_i$
such that $C^*_i \rhd_i G_i$
(note that, if $G_i = \emptyset$, then $C^*_i = C_i$).
An importance ordering $\rhd_i$ is \emph{separable} if
either $C_i \rhd_i G_i$ or $G_i \rhd_i C_i$. 
In other words, in separable ordering
either all chores are terrible or every good is more important than every chore.

\paragraph{Instance.}
An \emph{instance of the allocation problem with mixed items} (a mixed instance)
is a four-tuple $(N,M,G,\rhd)$, where
$G \coloneqq (G_i)_{i\in N}$ and $\rhd \coloneqq (\rhd_i)_{i\in N}$.
An instance is \textit{goods-only} if $G = M$,  \textit{chores-only} if $G = \emptyset$, and is \textit{objective} if $G_{i} = G_{j}$ for every $i,j\in N$. 
%
An instance $(N,M,G,\rhd)$ is separable if
the importance orderings are separable.
Note that separable instances can be seen as a special extension of lexicographic preferences over mixed items with the assumption that for every agent
either chores are more important than goods
or goods than chores.
On the other hand, lexicographic preferences can be seen as a special case of additive preferences in which the magnitude of valuations grow exponentially in the importance ordering.
Figure~\ref{fig:classes} illustrate the inclusion relation between different lexicographic extensions.

A \emph{mixed instance with terrible chores} is a (possibly objective) instance
in which $C^*_i \neq \emptyset$ for every $i \in N$.
The following example illustrates such an instance.

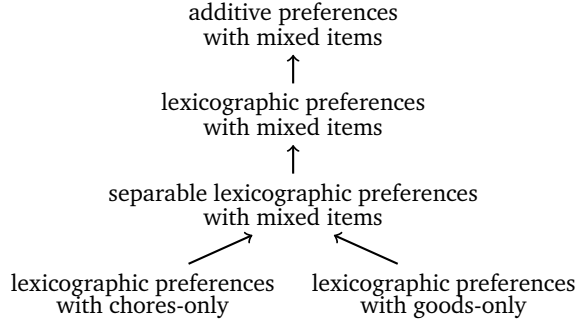
\begin{figure}[t]
    \centering
    \begin{tikzpicture}\small
      \tikzset{
        blank/.style={minimum size=0.001cm},
        edge/.style={->,draw,thick}
      }
      \def\sx{4cm} 
      \def\sy{1.2cm} 
      \def\x{0cm}
      \def\y{0.3cm}
      \def\txtsize{\footnotesize}
      \node[blank] (C) at (\x + 0*\sx, 0*\sy + \y) {\txtsize lexicographic preferences};
      \node[blank] (C_) at (\x + 0*\sx, 0*\sy) {\txtsize with chores-only};
      \node[blank] (G) at (\x + 1*\sx, 0*\sy + \y) {\txtsize lexicographic preferences};
      \node[blank] (G_) at (\x + 1*\sx, 0*\sy) {\txtsize with goods-only};
      \node[blank] (SepLex) at (\x + 0.5*\sx, 1*\sy + \y) {\txtsize separable lexicographic preferences};
      \node[blank] (SepLex_) at (\x + 0.5*\sx, 1*\sy) {\txtsize with mixed items};
      \node[blank] (Lex) at (\x + 0.5*\sx, 2*\sy + \y) {\txtsize lexicographic preferences};
      \node[blank] (Lex_) at (\x + 0.5*\sx, 2*\sy) {\txtsize with mixed items};
      \node[blank] (Add) at (\x + 0.5*\sx, 3*\sy + \y) {\txtsize additive preferences};
      \node[blank] (Add_) at (\x + 0.5*\sx, 3*\sy) {\txtsize with mixed items};
      
      \path[edge]
      (C) edge (SepLex_)
      (G) edge (SepLex_)
      (SepLex) edge (Lex_)
      (Lex) edge (Add_)
      ;
    \end{tikzpicture}
    \caption{Inclusion relation in different lexicographic extensions.}
    \label{fig:classes}
\end{figure}

%

\begin{example} \label{ex:efx+po}
    Consider a mixed instance with three agents, six items and a profile as follows. The set of common goods is $\bar{G} = \{o_{5}^{+}, o_{6}^{+}\}$, and the set of common chores is $\bar{C} = \{o_{1}^{-}, o_{2}^{-}\}$.
    This mixed instance contains terrible chores because every agent has a top item as a chore, i.e., $\rhd_i(1) \in C_i$.
    In fact, it is a separable instance as well.
    \begin{align*}
        1 &: \quad o_{1}^{-} \rhd \underline{o_{2}^{-}} \rhd o^{-}_3 \rhd \underline{o^{+}_4} \rhd o_{5}^{+} \rhd o_{6}^{+} \\
        2 &: \quad \underline{o_{1}^{-}} \rhd o_{2}^{-}  \rhd o^{-}_3 \rhd o^{-}_4 \rhd  \underline{o_{5}^{+}} \rhd \underline{o_{6}^{+}}  \\
        3 &: \quad o_{1}^{-} \rhd o_{2}^{-}  \rhd \underline{o^{+}_3} \rhd o^{+}_4 \rhd  o_{5}^{+} \rhd o_{6}^{+}
    \end{align*}
    (underline denotes an allocation described in \cref{sec:separable}).
\end{example}

\paragraph{Allocations.}
An allocation $A \coloneqq (A_i)_{i \in N}$ is a partition of $M$ such that $A_i \subseteq M$ is agent $i$'s bundle.
An allocation is \emph{complete} if all items in $M$ are assigned,
i.e., $\bigcup_{i\in N}  A_i = M$ and is \emph{partial} otherwise. Unless explicitly stated, we assume that an allocation is complete.

\paragraph{Envy-freeness.}
Given a pair of agents $i, j \in N$, agent $i$ \emph{envies} $j$ if
$A_j \succ_i A_i$.
Allocation $A$ is \emph{envy-free} (\EF{}), if for every pair of agents $i, j \in N$, we have $A_i \succeq_i A_j$.
Allocation $A$ is \emph{envy-free up to one item} (\EF{1}),
    if for every $i,j \in N$ such that $i$ envies $j$,
    there is $g \in G_i \cap A_j$ such that
    $A_i \succeq_i A_j \setminus \{g\}$ or
    there is $c \in C_i \cap A_i$ such that
    $A_i \setminus \{c\} \succeq_i A_j$.
Allocation $A$ is \emph{envy-free up to any item} (\EFX{}),
    if for every $i,j \in N$ such that $i$ envies $j$, it holds that 
    for every $g \in G_i \cap A_j$ we have
    $A_i \succeq_i A_j \setminus \{g\}$ and
    for every $c \in C_i \cap A_i$ we have
    $A_i \setminus \{c\} \succeq_i A_j$.






\paragraph{Maximin Share.}
The \emph{maximin share} (\MMS{}) of an agent is the most preferred bundle it can guarantee by creating an $n$-partition and receiving the worst one. Formally, for agent $i\in N$,
\(
    \MMS{}_i \coloneqq \max_{A \in \Pi_n} \min \{A_1,\dots,A_n\},
\)
where $\Pi_n$ is the set of all $n$-partitions of $M$, and $\max$ and $\min$ denote the most preferred
and the least preferred bundles according to $\succ_i$, respectively.
An allocation $A$ satisfies \emph{maximin share}, if for every $i \in N$ it holds that $A_i \succeq_i \MMS{}_i$.
In our setting \EFX{} implies both \EF{1} and \MMS{}, but the converse is not true. Moreover, \EF{1} and \MMS{} do not imply each other.




\paragraph{Economic  Efficiency.}
A (possibly partial) allocation $A$ \emph{Pareto dominates} allocation $B$
if $A$ assigns the same set of items as $B$, i.e.,
$\bigcup_{i\in N}  A_i = \bigcup_{i\in N}  B_i$, and
$A_i \succeq_i B_i$ for every $i \in N$
and there exists $i \in N$ such that $A_i \succ_i B_i$.
Allocation $A$ is \emph{Pareto optimal} (PO), if it is not Pareto dominated by any other allocation. 
In Section~\ref{sec:mms+po},
we also consider \emph{rank maximality} (RM),
which is a stronger efficiency notion.
Intuitively, it means that each item is given to an agent
that values it the most.
We give a formal definition in Appendix~\ref{app:mms+rm}.

We note that \EFX{} and \PO{} allocations always exist
and can be efficiently computed in every instance
without terrible chores,
i.e., when at least one agent sees its most important item as a good  \citep{hosseini2022fairly}.
Thus, we primarily focus on instances with terrible chores.

\paragraph{Serial Dictatorship.}
An \emph{ordering} of agents is a sequence
$\sigma  = (\sigma_1, \dots , \sigma_n)$ such that
$\sigma_i \in N$ denotes the $i$-th agent in the sequence.
A \emph{quota} vector $q = (q_1, \dots, q_n)$
is a vector of integers that we assign to each agent.
A \emph{serial dictatorship} mechanism, prescribed by an ordering $\sigma$ and a quota $q$, proceeds as follows:
starting from some partial (possibly empty) allocation $A$,
in each step, $i \in [k]$,
if there are still unallocated items,
we take the $q_{\sigma_i}$ most preferred ones by agent $\sigma_i$
(i.e., the most important goods
and then the least important chores if there are not enough goods left)
and add it to the bundle of this agent.


\section[Envy-Freeness up to Any Item (EFX)]{Envy-Freeness up to Any Item (\EFX{})} \label{sec:EFX}
Recall that an \EFX{} allocation may fail to exist for mixed instances \citep{hosseini2022fairly}. This non-existence crucially relies on a set of common chores with the highest importance ordering, i.e., the terrible chores. Otherwise, if there is at least one agent with the top item as a good, an \EFX{} and \PO{} allocation can be computed efficiently under lexicographic preferences.
Thus, we focus on mixed instances with terrible chores and show that deciding whether an \EFX{} allocation exists is computationally hard under lexicographic preferences, which subsequently implies hardness for additive preferences with mixed items.
The formal proof is relegated to \cref{app:efx:hardness}.



\begin{restatable}{theorem}{thmefxhradness}
\label{thrm:efx:hardness}
Deciding whether there exists an \EFX{} allocation
for a given lexicographic mixed instance is NP-complete.
\end{restatable}

\begin{proof}[Proof (sketch).]
We prove the hardness by a reduction from \textsc{Exact Cover by 3-Sets} (\textsc{X3C}).
In an \textsc{X3C} instance, we have a universe
$\mathcal{U} = \{u_1,\dots,u_{m}\}$ and
a family of its three-element-subsets,
$\mathcal{S} = \{S_1,\dots,S_n\}$.
The problem, which is known to be NP-complete~\citep{johnson1979computers},
is to decide whether there exists an exact cover
$\mathcal{K} \subset \mathcal{S}$ of size $k$,
such that $\bigcup_{S_j \in \mathcal{K}} S_j = \mathcal{U}$.

For every such \textsc{X3C} instance,
we construct a corresponding
objective mixed items instance $(N,M,G,\rhd)$ as follows.
For every element $u_i \in \mathcal{U}$,
we take $2n$ common chores $c_{i,1},\dots,c_{i,2n}$.
We add to it $k$ common goods $g_1,\dots,g_k$
(the assignment of which will correspond to
the choice of subsets in $\mathcal{K}$),
which gives us $|M| = 2mn + k$ items in total.
Next, for every subset $S_j \in \mathcal{S}$
we take two agents $2j-1$ and $2j$
and we give them identical importance orderings,
i.e., $\rhd_{2j-1} = \rhd_{2j}$.
Specifically, their importance ordering consists of
three ``blocks'':
first there are all chores corresponding to elements $u_i$
such that $u_i \not \in S_j$,
then there are $k$ goods $g_1,\dots,g_k$,
and at the end there are all chores corresponding to elements $u_i \in S_j$.
For example, if we had $S_j = \{u_1,u_2,u_3\}$,
then the importance ordering of agents $2j-1$ and $2j$ would be
\(
    c_{4,1} \rhd \dots \rhd c_{m,2n} \rhd g_1 \rhd \dots \rhd g_k \rhd c_{1,1} \rhd \dots \rhd c_{3,2n}.
\)

We prove that there is a cover in an \textsc{X3C} instance,
if and only if,
there is an EFX allocation in the corresponding mixed item instance.
If there is a cover, without loss of generality we assume that $\mathcal{K} = \{S_1,\dots,S_k\}$
and show that allocation $A=(A_1,\dots,A_{2n})$, where
\begin{equation}
\label{eq:efx:hardness}
    A_{j} \!= \!\begin{cases}
        \!\{g_{j/2}\} \!\cup\! \{c_{i,l} \!\in\! \bar{C} : u_i \!\in\! S_{j/2}\}, & \!\!\!\mbox{if } j \!\in\! \{2,\! ...,2k\},\\
        \emptyset, & \!\!\!\mbox{otherwise,}
    \end{cases}
\end{equation}
is EFX (and also PO).
If there is no cover, we analyze the number of ``uncovered'' chores in an allocation,
i.e., chores received by an agent that are more important than every good it received.
We show that EFX would imply that there can be at most $2n-1$ such chores,
but no set cover implies that there are at least $2n$ of them---a contradiction.
\end{proof}

In the proof of Theorem~\ref{thrm:efx:hardness}, we show that
the allocation defined in equation~\eqref{eq:efx:hardness}
(i.e., an \EFX{} allocation that exists
when there is a set cover in an \textsc{X3C} instance)
is not only \EFX{} but also \PO{}.
This implies that deciding whether there exists
an allocation that is both \EFX{} and \PO{} is also \NPH{}
(we note that the problem of verifying if an allocation is PO in polynomial time remains open, thus we cannot claim NP-completeness).
\begin{corollary}
\label{cor:efx+po:hardness}
Deciding whether there exists an \EFX{} and PO allocation
for a given lexicographic mixed instance is NP-hard.
\end{corollary}

The constructions in the proof of \cref{thrm:efx:hardness} only used objective instances, where all agents agree on whether an item is a chore or a good. Thus, these computational hardness results hold for all mixed instances and do not rely on subjective views of agents.

\subsection[EFX and PO: Separable Preferences]{\EFX{} and \PO{}: Separable Preferences} \label{sec:separable}

An important feature of our construction in the proof
of Theorem~\ref{thrm:efx:hardness}
is that each agent has some terrible chores
and some other (non-terrible) chores
that are separated by several goods in its importance ordering.
%
%
%
In this section, we analyze the case where
either all chores are terrible or
all are less important than every good,
i.e., the separable lexicographic preferences.
We show that such a constraint
enables us to devise an algorithm that
computes an \EFX{} and \PO{} allocation
for every separable instance.

Algorithm~\ref{alg:efx+po}
finds one such allocation
for every mixed instance with separable preferences
that contains terrible chores.
It extends the algorithm by \citet{hosseini2022fairly} for \EFX{} and \PO{} allocations in instances in which the most important item of one of the agents is a good (in fact Phase 2 of our algorithm can be seen as running this algorithm on a smaller instance).

\paragraph{Algorithm.}
Fix any ordering of agents $1, \dots, n$.
The algorithm runs in two phases. In Phase 1, we allocate all common chores (items in $\bar{C}$) and goods to agents that receive these chores.
First, all common chores are allocated through a serial dictatorship with ordering $(1,\dots,n)$ and quotas $q$, where $q_i = 1 $ for $i \in N$,
except $q_1 = |\bar{C}| - n + 1$, if $|\bar{C}| > n$.
Then starting from the last agent which received a common chore (i.e., the worst chore), say agent $z$, in the reverse order, i.e., $z,z-1,\dots,1$,
we add to each agent's bundle all the unassigned items that it considers as goods.
At the end of Phase 1, all remaining items are considered as a good for at least one agent in $N$, but considered as a chore by all agents who received an item in Phase 1, i.e., agents in $[z]$.
This is crucial for ensuring that the final allocation will be \EFX{}.
In Phase 2, we distribute the remaining items in such a way
that each is assigned to an agent for which it is a good.
Specifically, we move through the positions in importance orderings,
one by one, starting from the first position with an unassigned item.
For each position $k$, we find an agent $i \in N'$,
that has yet unassigned good at position $k$, and assign this good plus all remaining items that only $i$ considers as goods (but no other remaining agent considers as goods).
This process repeats until no item remains unassigned. 
Example~\ref{ex:separable} illustrates the execution of Algorithm~\ref{alg:efx+po}).


\begin{example}\label{ex:separable}
    Consider the mixed instance with separable preferences given in \cref{ex:efx+po}. 

    \cref{alg:efx+po} starts by running a serial dictatorship with a fixed ordering of $(1,2,3)$ to allocate all terrible common chores, i.e., $\{o_{1}^{-}, o_{2}^{-}\}$. 
    Thus, agents $1$ and $2$ receive $o_{2}^{-}$ and $o_{1}^{-}$ respectively, while agent $3$ receives nothing.
    Then, starting from agent $2$, the last agent who received a chore (the worst chore), in the reverse ordering i.e., $(2,1)$, agents receive all their remaining goods (if any). Therefore, at the end of Phase 1, $A_1 = \{o^-_2, o^+_4\}$, $A_2 = \{o^-_1, o^+_5, o^+_6\}$ and $A_3 = \emptyset$.
    
    In Phase 2, the only remaining item $o_3$ is allocated to agent $3$ who sees it as a good (while agents 1 and 2 consider $o_3$ as a chore). 
    The final allocation is underlined in \cref{ex:efx+po}.

\end{example}

\begin{algorithm}[t]\small
    \caption{Computing an \EFX{} and \PO{} allocation for separable lexicographic preferences}
    \label{alg:efx+po}
    \begin{algorithmic}[1]
        \REQUIRE A mixed instance $( N, M, G, \rhd)$ with separable preferences that contains terrible chores
        \ENSURE An allocation $A$ that is \EFX{} and PO
        \STATE $A := (\emptyset, \dots, \emptyset)$
        
        $\triangleright$ \textsc{Phase 1:}
        \STATE assign common chores, $\bar{C}$, by a serial dictatorship
        with ordering $(1,\dots,n)$ and quotas $q$,
        where $q_1 = \max(1,|\bar{C}|-n+1)$ and $q_i = 1$, for $i \in N \setminus \{1\}$ 
        \STATE $N' := \{i \in  N: A_i = \emptyset \}, \quad H := M \setminus \bar{C}$
        \STATE $z := \max\{N \setminus N'\}$
        \FOR{$i \in (z, z-1, \dots, 1)$}
            \STATE $A_i \leftarrow A_i \cup (G_i \cap H), \quad H \leftarrow H \setminus A_i$
        \ENDFOR
        
        $\triangleright$ \textsc{Phase 2:}
        \FOR{$k \in (1,\dots,|M|)$} 
            \WHILE{there is $i \in N'$ such that $\rhd_i(k) \in H \cap G_i$}
                \STATE $g_i := \rhd_i(k), \quad A_i \leftarrow \{g_i\}$
                \STATE $A_i \leftarrow A_i \cup \big( H \cap G_i \setminus \bigcup_{j \in N' \setminus \{i\}} G_j \big)$
                \STATE $N' \leftarrow N' \setminus \{i\}, \quad H \leftarrow H \setminus A_{i}$
            \ENDWHILE
        \ENDFOR
            \RETURN $A$ 
        \end{algorithmic}
\end{algorithm}
Before proving the correctness of the algorithm,
let us show a general result (Lemma~\ref{lmm:po})
concerning the serial dictatorship mechanism.
Assume that there is a PO partial allocation
and the set of common chores $H$ such that
for each agent $H$ is more important than its bundle
and all its goods (i.e., $H$ contains terrible chores for all agents).
We show that extending such partial allocation
by allocating items in $H$ through the serial dictatorship
with arbitrary ordering and quotas
will preserve PO.
Since the order of allocating items does not affect the final allocation, \cref{lmm:po} can be used to show the correctness of \cref{alg:efx+po}.

\begin{lemma}
\label{lmm:po}
    For every instance $(N, M, G, \rhd)$, subset of common chores $H \subseteq \bar{C}$, 
    and partial allocation $B$ of items in $M \setminus H$ that is PO,
    if for each $i \in N$ it holds that $H \rhd_i (G_i \cup B_i)$,
    then every allocation $A$ obtained by
    extending $B$ by the serial dictatorship
    with an arbitrary ordering and quotas is PO.
\end{lemma}
\begin{proof}
    Assume by contradiction that there exists
    $A$ obtained by the serial dictatorship
    that is not PO.
    This means that there exists an allocation $A'$
    that Pareto dominates $A$.
    Now, let us consider two cases
    based on whether $A$ and $A'$ differ on assignment of chores in $H$.
    If this is true, then there exists a common chore, $c \in H$,
    that is assigned to different agents in $A$ and $A'$,
    i.e., there exists $i \in N$ such that $c \in A'_i \setminus A_i$.
    Let us take $c$ and $i$ such that, among all such chores, $c$ 
    was picked as the last one in the serial dictatorship leading to $A$.
    Observe that every chore $c' \in H \cap A_i$ such that $c' \rhd_i c$
    was picked by agent $i$ after $c$ was assigned
    (otherwise $i$ would pick $c$ instead).
    Hence, for every such $c'$ we have also $c' \in A'_i$
    (otherwise $c$ would not be the last picked chore
    that is assigned to different agents in $A$ and $A'$).
    Since $H \rhd_i (G_i \cup B_i)$,
    this implies that $A_i \succ_i A'_i$,
    which means that $A'$ does not Pareto dominate $A$---a contradiction.

    Finally, consider the case in which $A$ and $A'$
    assign chores in $H$ identically.
    By $B'$ let us denote the partial allocation obtained from $A'$
    by removing chores in $H$.
    Since $A'$ Pareto dominates $A$, it means that $B'$ Pareto dominates $B$.
    But that contradicts the fact that $B$ is PO.
\end{proof}


\begin{restatable}{theorem}{thrmefxposep}
\label{thrm:efx+po+sep}
Given a mixed instance with separable preferences that contains terrible chores, an \EFX{} and \PO{} allocation always exists and can be computed in polynomial time.
\end{restatable}


The proof is relegated to \cref{app:efx+po:separable}.
%
When at least one agent's top item is a good, an \EFX{} and \PO{} allocations are guaranteed to exist and can be computed efficiently \citep{hosseini2022fairly}.
Combining this with \cref{thrm:efx+po+sep} we obtain the following
computational and existence results for separable preferences.
\begin{corollary}
    \label{cor:efx+po}
    Given any mixed instance with separable preferences, an \EFX{} and \PO{} allocation always exists and
     can be computed in polynomial time.
\end{corollary}


\section[EF1 and PO]{EF1 and \PO{}} 
\label{sec:EF1+PO}

Despite the non-existence of \EFX{} and the computational hardness of deciding whether an instance admits such an allocation (\cref{thrm:efx:hardness}),
we identified a natural class of separable lexicographic preferences for which an \EFX{} and \PO{} allocation is always guaranteed to exist and can be computed efficiently (\cref{thrm:efx+po+sep}). 
This raises the question of whether focusing on weaker fairness notions, e.g., \EF{1} or \MMS{}, enables us to escape these negative results for the more general mixed lexicographic (but not necessarily separable) preferences.

In this section, we focus on \EF{1} and discuss the technical challenges in satisfying it with \PO{}. We then devise an efficient algorithm for finding \EF{1} and \PO{} when there are sufficiently many \textit{common} terrible chores, in particular, when there are at least $n-1$ terrible chores shared by all agents.

%


%


Before presenting our main result in this section, let us discuss the technical challenges in achieving \EF{1} and \PO{} for mixed lexicographic preferences.\footnote{The existence and computation of \EF{1} and \PO{} allocations remain open even for additive chores-only instances.}
For mixed instances, under additive or doubly monotone valuations,\footnote{Doubly (or item-wise~\citep{chen2020fairness}) monotone is a broad valuation class wherein each agent can partition items into those with non-negative (goods) or negative (chores) marginal utility.} an \EF{1} allocation (without \PO{}) can be computed efficiently through either the double round robin algorithm \citep{ACI+19fair} or a variant of the envy-graph algorithm \citep{BSV21approximate}.
However, both these approaches fail in satisfying \PO{} even when preferences are restricted to the lexicographic domain, as we illustrate in \cref{app:ef1+po:examples}.

\begin{remark}
Other naive approaches 
also fail to achieve the desired outcome. 
For instance, a 
good may have to be assigned to an agent
for which it is not in the highest position.
This observation immediately shows that approaches used for achieving \EFX{} and \PO{} under separable preferences (as described in \cref{alg:efx+po}) or those proposed by \citep{hosseini2022fairly} 
that assign goods to agents having them high in the orderings.
We illustrate this challenge in the next example.
\end{remark}

\begin{example}
\label{ex2:ef1+po}
 Consider a mixed instance with five agents, six items, and a profile as follows. The set of common goods is $\bar{G} = \{o_{1}^{+}\}$, the set of common chores is $\bar{C} = \{o_{2}^{-}, o_{3}^{-}, o_{4}^{-}, o_{5}^{-}, o_{6}^{-}\}$, but the set of common terrible chores is empty.
    \begin{align*}
       1 &: \quad o^-_4 \rhd o^-_5 \rhd o^-_6 \rhd o^+_1 \rhd o^-_2 \rhd o^-_3 \\
       2 &: \quad o^-_2 \rhd o^-_3 \rhd o^+_1 \rhd o^-_4 \rhd o^-_5 \rhd o^-_6 \\
       3 &: \quad o^-_2 \rhd o^-_3 \rhd o^+_1 \rhd o^-_4 \rhd o^-_5 \rhd o^-_6 \\
       4 &: \quad o^-_2 \rhd o^-_3 \rhd o^+_1 \rhd o^-_4 \rhd o^-_5 \rhd o^-_6 \\
       5 &: \quad o^-_2 \rhd o^-_3 \rhd o^+_1 \rhd o^-_4 \rhd o^-_5 \rhd o^-_6 
    \end{align*}
In this instance, every \EF{1} and \PO{} allocation must assign the only good, $o_{1}^{+}$ to agent 1; otherwise, either the allocation violates \PO{} or it violates \EF{1}. Note that all other agents rank $o_{1}^{+}$ higher in their importance ranking; yet, this common good must be allocated to agent 1.

Another unintuitive observation is that if the preferences of agents 3, 4, and 5 were identical to those of agent 1 (instead of 2), then $o_{1}^{+}$ would need to be allocated to agent 2 to guarantee \EF{1} and \PO{}.
In \cref{app:ef1+po:examples}, we give additional examples to illustrate the complexity of this problem.
\end{example}

Given the aforementioned challenges, we show that for lexicographic mixed instances that contain at least $n-1$ common terrible chores, an \EF{1} and \PO{} allocation always exists and can be computed efficiently.
Formally, the set of \emph{common terrible chores} contains all chores that are terrible for all agents, i.e., $\bar{C}^* = \bigcap_{i \in N} C^*_i$.
We describe an algorithm that finds an \EF{1} and \PO{} allocation for every mixed instance with at least $n-1$ common terrible chores, i.e., $|\bar{C}^*| \geq n - 1$.
We present its pseudocode in \cref{app:ef1+po:terrible}.

\paragraph{Algorithm.}
Fix any ordering of agents $1, \ldots, n$. 
We start by giving agent $1$ all items it considers as goods.
To each next agent, in the order $2,\dots,n$,
we give everything it considers as goods
from the set of unassigned items
(or nothing if there are no such items left).
The remaining items are necessarily common chores.
Next, we start from agent $n$,
and assign to it all of its non-terrible chores.
To each next agent, in the reversed order, i.e., $n-1,\dots,1$,
we give all its non-terrible chores
that remain (if any).
The only remaining items are common terrible chores.
Such partial allocation is \PO{}, but it can be very unfair
(agent $1$ got all its goods and
agent $n$ all its non-terrible chores).
To ensure fairness, we assign the remaining
common terrible chores using serial dictatorship
with ordering $\sigma$ such that the last agent, $\sigma_n$,
is not envied by any other agent
(since the partial allocation is \PO{}
there surely is such $\sigma$).
In this way, every agent
(except possibly $\sigma_n$)
receives at least one common terrible chore,
which results in an \EF{1} allocation
(and by \cref{lmm:po} it is still \PO{}).

\begin{algorithm}[t]
	\caption{Finding an \EF{1} and PO allocation when there is at least $n-1$ common terrible chores}
    \label{alg:ef1+po:n-1}
    \begin{algorithmic}[1]
    \REQUIRE A mixed instance $( N, M, G, \rhd )$ s.t. $|\bar{C}^*| \ge n -1$
    \ENSURE An allocation $A$ that is \EF{1} and PO 
        \STATE $A := (\emptyset, \dots, \emptyset), \quad H := \bar{C} \setminus \bar{C}^*$

        $\triangleright$ \textsc{Phase 1:}
        \FOR{ $i \in (1,\dots,n)$}
            \STATE $A_i \leftarrow G_i \setminus \bigcup_{j \in [i-1]} A_j$ 
        \ENDFOR
        \FOR{ $i \in (n,\dots,1)$}
            \STATE $A_i \leftarrow A_i \cup H \setminus C^*_i$
            \STATE $H \leftarrow H \setminus A_i$
        \ENDFOR

        $\triangleright$ \textsc{Phase 2:}
        \STATE take an arbitrary ordering $\sigma$ s.t. no one envies agent $\sigma_n$
        \STATE assign the items in $\bar{C}^*$ by the serial dictatorship with ordering $\sigma$ and quotas $q$,
            where $q_{\sigma_1} = \max(1,|\bar{C}^*|-n+1)$ and $q_{\sigma_i} = 1$, for every $i \in \{2,\dots,n\}$
        \RETURN $A$
    \end{algorithmic}
\end{algorithm}

\begin{example}
    \label{ex:ef1+po}
    Consider a mixed instance with three agents, eight items, and a profile as follows. 
    The set of common chores is $\bar{C} = \{o^-_1, o^-_2, o^-_3, o^-_5\}$. 
    \begin{align*}
       1 &: \quad o^-_1 \rhd o^-_2 \rhd \underline{o^-_3} \rhd \underline{o^+_4} \rhd o^-_5 \rhd \underline{o^+_6} \rhd o^{-}_7 \rhd \underline{o^+_8} \\
       2 &: \quad o^-_1 \rhd \underline{o^-_2} \rhd o^-_3 \rhd o^+_4 \rhd o^-_5 \rhd o^+_6 \rhd \underline{o^{+}_7} \rhd o^+_8 \\ 
       3 &: \quad \underline{o^-_1} \rhd o^-_2 \rhd o^-_3 \rhd o^+_4 \rhd \underline{o^-_5} \rhd o^+_6 \rhd o^{-}_7 \rhd o^+_8
    \end{align*}
         Suppose the ordering is $(1,2,3)$.
         \cref{alg:ef1+po:n-1} starts by assigning $\{o^+_4, o^+_6, o^+_8\}$ and $\{o^+_7\}$ to agents $1$ and $2$, respectively.
         Then in the reverse ordering $(3,2,1)$, agents get their common non-terrible chores (out of the remaining items), resulting in agent $3$ receiving $o^-_5$ (and nothing for others).
         Since agent $3$ is not envied (such an agent always exists), \cref{alg:ef1+po:n-1} allocates all common terrible chores ($\{o^-_1, o^-_2, o^-_3\}$) by running a serial dictatorship with the ordering of $(1,2,3)$ and single quota. The final allocation is underlined. 
%
\end{example}

\begin{restatable}{theorem}{thmefonePOterrible}
\label{thm:ef1POterrible}
    Given a lexicographic mixed instance with at least $n-1$ common terrible chores, an \EF{1} and \PO{} allocation always exists and can be computed in polynomial time.
\end{restatable}


The formal proof is relegated to \cref{app:ef1+po:terrible}.
Given the theorem above, one may wonder whether a similar approach can be utilized for instances with potentially fewer than $n-1$ common terrible chores. 
In \cref{app:ef1+po:examples}, we show that even extending to $n-2$ (if possible) requires new techniques with a rather complicated analysis to guarantee Pareto optimality.

\section[MMS and Efficiency]{\MMS{} and Efficiency}
\label{sec:mms+po}

Despite the challenges in satisfying \EF{1} and \PO{} for lexicographic mixed instances that contain terrible chores, we show that an \MMS{} and \PO{} allocation always exists and can be computed in polynomial time. 
%
Note that while an \MMS{} allocation can be computed efficiently \citep{hosseini2022fairly}, its computation along with economic efficiency notions such as \PO{} and rank maximality was open even for objective instances.





We build on the characterization of maximin share which is specified by the most important items.
Simply put, an agent's \MMS{} is characterized by its top item: if it is a chore, its \MMS{} is the top item and all its goods; otherwise its \MMS{} is the set of all items that it considers as good without the first $n-1$ goods according to its importance ordering (or the empty set if its importance ordering contains fewer than $n$ goods).

\begin{proposition}
\label{prop:mms_threshold}
\emph{\citep{hosseini2022fairly}}
Given a mixed instance $(N,M,G,\rhd)$, for every agent $i \in N$,
if $\rhd_i(1) \in C_i$, it holds that $\MMS{}_i = \{\rhd_i(1)\} \cup G_i$.
Otherwise, if $\rhd_i(1) \in G_i$, it holds that $\MMS{}_i = \emptyset$, if $|G_i| < n$,
or $\MMS{}_i = G_i \setminus \bigcup_{k \in [n-1]} \{\rhd_i(k,G_i)\}$, if $|G_i| \ge n$.
\end{proposition}


\paragraph{Algorithm.}
Fix any ordering of agents $1,\dots, n$.
Similar to \cref{alg:ef1+po:n-1},
we start by allocating to 
each agent, in the ordering $1,\dots,n$,
all remaining unassigned items
that it considers as goods
(or nothing if there are no such items left).
The remaining items will be common chores.
We give all of them to agent $n$,
with the exception of
the most important item for $n$,
which we denote by $c^*$.
Now, the choice of which agent should receive $c^*$
depends on whether $c^*$ is the
most important item for all agents.
If it is the case, we give it to agent $1$.
Otherwise,
i.e., if there is at least one agent
for which there is more important item than $c^*$,
we give it to the last such agent
in the ordering.
\begin{example}
    \label{ex:mms+po}
    We revisit the instance given in Example~\ref{ex:ef1+po}. 
    \begin{align*}
       1 &: \quad \underline{o^-_1} \rhd o^-_2 \rhd o^-_3 \rhd \underline{o^+_4} \rhd o^-_5 \rhd \underline{o^+_6} \rhd o^{-}_7 \rhd \underline{o^+_8} \\
       2 &: \quad o^-_1 \rhd o^-_2 \rhd o^-_3 \rhd o^+_4 \rhd o^-_5 \rhd o^+_6 \rhd \underline{o^{+}_7} \rhd o^+_8 \\ 
       3 &: \quad o^-_1 \rhd \underline{o^-_2} \rhd \underline{o^-_3} \rhd o^+_4 \rhd \underline{o^-_5} \rhd o^+_6 \rhd o^{-}_7 \rhd o^+_8
    \end{align*}
    For this instance, the allocation returned by Algorithm~\ref{alg:ef1+po:n-1} is not \MMS{}. The outcome for agent $3$ was $\{o^-_1, o^-_5\}$, to which agent $3$ strictly prefers its \MMS{} $(\MMS{}_3 = \{o^-_1, o^+_4, o^+_6, o^+_8\}$).
    Suppose the ordering is $(1,2,3)$. Algorithm~\ref{alg:mms+po}  starts by assigning $\{o^+_4, o^+_6, o^+_8\}$ and $\{o^+_7\}$ to agents $1$ and $2$, respectively.
    Then, agent $3$ receives all common chores, except its most important item $c^*=o^-_1$. Lastly, since $c^*$ is the most important item for every agent, it is allocated to the first agent. The final allocation is underlined. 

    
\end{example}

Let us prove the correctness of our algorithm
(the full proof is relegated to \cref{app:mms+po}).
\begin{algorithm}[t]\small
    \caption{Finding an \MMS{} and PO allocation}
    \label{alg:mms+po}
    \begin{algorithmic}[1]
        \REQUIRE A mixed instance $( N, M, G, \rhd)$ with terrible chores
        \ENSURE An allocation $A$ that is \MMS{} and PO
        \STATE $A := (\emptyset, \dots , \emptyset), \quad c^* := \rhd_{n}(1)$
        \FOR{$i \in (1,\dots,n)$}
            \STATE $A_i \leftarrow G_i \setminus \bigcup_{j \in [i-1]} A_j$
        \ENDFOR
        \STATE $A_n \leftarrow A_n \cup \bar{C} \setminus \{c^*\}$
        \IF{for every $i \in N$ it holds that $\rhd_i(1) = c^*$}
            \STATE $A_1 \leftarrow A_1 \cup \{c^*\}$
        \ELSIF{$c^* \in \bar{C}$}
            \STATE $i^* := \max \{i \in [n]: \rhd_i(1) \neq c^* \}$
            \STATE $A_{i^*} \leftarrow A_{i^*} \cup \{c^*\}$
        \ENDIF
        \RETURN $A$
    \end{algorithmic}
\end{algorithm}

\begin{restatable}{theorem}{thrmmmspo}
\label{thrm:mms+po}
    Given a lexicographic mixed instance with terrible chores, an \MMS{} and \PO{} allocation always exists and can be computed in polynomial time.
\end{restatable}
\begin{proof}[Proof (sketch)]
    Since $(N,M,G,\rhd)$ is an instance with terrible chores,
    by \cref{prop:mms_threshold},
    maximin share of every agent
    consists of its most important chore
    and all goods.
    The first agent, agent $1$, is the only one that can receive its most important chore
    in our algorithm.
    However, since apart from that it receives all of its goods,
    the output allocation is \MMS{}.

    
    For \PO{}, consider two agents $i < j \in [n]$.
    Observe that $j$ does not have any item
    that $i$ considers as good.
    Hence, the only Pareto improvement between those two agents
    is possible if $i$ received $c^*$
    (Pareto improvement can involve more than two agents,
    but we do not consider such in this sketch).
    Then, $i$ can potentially offer $c^*$ to $j$,
    bundled with less important for $i$ goods.
    However, if $i$ was assigned $c^*$,
    this means that $c^*$ is the most important item for $j$.
    Hence, $j$ would not accept any
    exchange in result of which it gets $c^*$.
\end{proof}

Combining Theorem~\ref{thrm:mms+po} with the existence and computation results when there are no terrible chores~\citep{hosseini2022fairly}, we obtain the following general conclusion. 

\begin{corollary}
    \label{cor:mms+po}
    Given a lexicographic mixed instance, an \MMS{} and \PO{} allocation always exists and can be computed in polynomial time.
\end{corollary}


Corollary~\ref{cor:mms+po} ensures that an \MMS{} and PO allocation always exists.
From Corollary~\ref{cor:efx+po:hardness} we know however
that if we strengthen \MMS{} to \EFX{},
then an \EFX{} and \PO{} allocation may not exist and deciding if such an allocation exists is computationally hard.
A natural question is whether one can strengthen the efficiency to rank maximality. We show that deciding whether there exists an \MMS{} and \RM{} allocation is computationally hard, which stands in sharp contrast to the goods-only and chores-only settings.


\begin{restatable}{theorem}{thrmmmsrm}
    \label{thrm:mms+rm}
Deciding whether there exists an \MMS{} and \RM{} allocation
for a given lexicographic mixed instance is \NPC{}.
\end{restatable}

The proof of the theorem (relegated to Appendix~\ref{app:mms+rm})
is a reduction from \textsc{Set Cover} problem and
shares some similarities with the proof of Theorem~\ref{thrm:efx:hardness}
(instance is objective and chores correspond to elements of the universe,
agents to subsets, and assignment of goods to subsets chosen to the cover).
However, there are some important differences.
First, much more emphasis is put on the positions of items in the importance orderings.
To this end, we introduce additional dummy goods and chores and two agents that
allow us to restrict the set of possible rank maximal allocations.

\section{Concluding Remarks}
By focusing on the restricted domain of lexicographic preferences, we identified instances with terrible chores for which \EFX{} is hard to compute, thus, providing the first ever computational hardness result for \EFX{}.
Nonetheless, we identified a natural class of separable lexicographic preferences for which \EFX{} and \PO{} allocations are efficiently computable (and always exist).
Moreover, we showed that \MMS{} and \PO{} allocations always exist and can be computed efficiently for any lexicographic mixed instance.

%

For \EF{1} and \PO{}, the main remaining challenge is how to deal with (possibly subjective) mixed instances that contain fewer than $n-1$ common terrible chores. 
%
Steps towards addressing this problem could potentially lead to novel techniques for more general preferences, including and beyond, the additive domain. 







\section*{Acknowledgments}
Hadi Hosseini acknowledges support from NSF grants \#2144413 (CAREER), \#2052488, and \linebreak \#2107173.
We thank the anonymous reviewers for their helpful comments.


\bibliographystyle{plainnat}
\bibliography{arxiv-23/mixedref,arxiv-23/references}


\appendix
\section*{Appendix}

\section[Omitted Material from Section 3]{Omitted Material from \cref{sec:EFX}}

\subsection[Proof of Theorem 1]{Proof of \cref{thrm:efx:hardness}}
\label{app:efx:hardness}

\thmefxhradness*

\begin{proof}
Since we can check whether given allocation is \EFX{} in polynomial time, we know that the problem is in NP.
Therefore, let us focus on proving the hardness.

To this end, we will make a reduction from \textsc{Exact Cover by 3-Sets} (\textsc{X3C}).
In an \textsc{X3C} instance, we are given a universe of $m=3k$ elements,
$\mathcal{U} = \{u_1,\dots,u_{m}\}$ and
a family of its three-element-subsets,
$\mathcal{S} = \{S_1,\dots,S_n\}$.
We will assume that each element of the universe
is contained in at least one of the subsets and
that each subset is distinct, i.e., $S_j \neq S_{j'}$ if $j \neq j'$.
The problem is to decide whether there exists a subfamily
$\mathcal{K} \subset \mathcal{S}$ of $k$ subsets
that is a set cover for $\mathcal{U}$, i.e.,
$\bigcup_{S_j \in \mathcal{K}} S_j = \mathcal{U}$.
Observe that since $m=3k$,
each element appears in exactly one subset in $\mathcal{K}$,
hence $\mathcal{K}$ is necessarily also an exact cover.
This problem is known to be NP-complete.

Now, for every \textsc{X3C} instance,
let us define a corresponding lexicographic mixed instance.
Specifically, let us take $2n$ agents $N= \{1,2,\dots,2n\}$.
Also, for every element $u_i \in \mathcal{U}$,
let us take $2n$ common chores $c_{i,1},c_{i,2},\dots,c_{i,2n}$.
This will give us a total of $2mn$ chores $C=\{c_{1,1},c_{1,2},\dots,c_{m,2n}\}$.
Let us add to it $k$ common goods $G = \{g_1,g_2,\dots,g_k\}$,
assignment of which will correspond to
the choice of $k$ subsets in the original instance.
This gives as a total of $k(6n + 1)$ items $M=C \cup G$.
Next, let us specify the importance ordering of each agent.
To this end, first, 
for every $j \in [n]$ by
$C^j = \{c_{i,l} \in C : u_i \in S_j\}$ 
let us denote the set of all chores corresponding
to the elements of subset $S_j$.
Now, for every $j \in [n]$,
we set the importance ordering of agent $2j$
such that the most important items for it are chores not in $C^j$,
followed by all goods, $G$,
and ending with chores in $C^j$, i.e.,
\[
    ( C \setminus C^j) \rhd_{2j} G \rhd_{2j} C^j.
\]
Moreover, we specify that the importance ordering of goods as
$g_1 \rhd_{2j} g_2 \rhd_{2j} \dots \rhd_{2j} g_k$.
The importance ordering of chores in $C \setminus C^j$ and $C^j$
can be arbitrary.
Finally, we set the importance ordering of agent $2j-1$
as identical to this of agent $2j$,
i.e., $\rhd_{2j - 1} = \rhd_{2j}$.
%
In the remainder of the proof we will show that
a set cover in the original instance exists,
if and only if, there exists an \EFX{} allocation in
the corresponding lexicographic mixed instance.

First, let us assume that there is a set cover $\mathcal{K} \subseteq \mathcal{S}$
in the original instance.
Let us show that this implies that there exists an \EFX{}
allocation in the corresponding lexicographic mixed instance
(on a side, we will also show that it is PO).
Without loss of generality,
let us assume that $\mathcal{K} = \{S_1,\dots,S_k\}$.
Then,
consider allocation $A=(A_1,\dots,A_{2n})$
such that for every $j \in [n]$
we have $A_{2j-1} = \emptyset$
and
\[
    A_{2j} = \begin{cases}
        \{g_j\} \cup C^j, & \mbox{if } j \le k,\\
        \emptyset, & \mbox{otherwise.}
    \end{cases}
\]
Since $\mathcal{K}$ covers all elements of $\mathcal{U}$ exactly once,
all items belong to some bundle in $A$,
which means it is a well-defined allocation.

Let us show that $A$ is also an \EFX{} allocation.
Consider an agent that received a non-empty bundle,
i.e., $2j$ for arbitrary $j \in [k]$.
Since its most important received item is a good,
it does not envy any agent whose bundle is an empty set.
Next, for every $j' \in [k] \setminus \{j\}$,
we know that $S_{j'} \neq  S_j$,
so also $C^{j'} \neq  C^j$.
Thus, there is a chore $c \in C^{j'}$ such that $c \not \in C^j$,
which means that $c \in A_{2j'}$ and $c$ is more important
for agent $2j$ than good $g_{j'}$.
Hence, for agent $2j$ the most important item in bundle $A_{2j'}$ is a chore.
Thus, $2j$ does not envy $2j'$.
As a result, we get that for every $j \in [k]$ agent $2j$
does not envy any other agent.
It remains to consider agents whose bundles are empty sets.
Observe that each such agent
can envy only the agents that received some items,
i.e., agents $2,4,\dots,2k$.
However, each of these agents received only one good,
after removal of which there is no more envy.
Therefore, $A$ is indeed \EFX{}.

On a side, let us observe that $A$ is also PO.
To prove that let us consider allocation $A'=(A'_1,\dots, A'_{2n})$
such that for each agent $j \in [2n]$ bundle $A'_j$
is at least as good as its original bundle $A_j$, i.e.,
\begin{equation}
\label{eq:thrm:efx:hardness:po}
\tag{$*$}
A'_j \succeq_i A_j, \quad \mbox{for every } j \in N.
\end{equation}
We will prove that this implies that $A'=A$.
First, let us show that for every $j \in [k]$,
it holds that $A'_{2j} \cap G = A_{2j} \cap G$,
i.e., that in $A'$ every $g_j$ still belongs to agent $2j$.
Assume otherwise and take the smallest $j$
such that $g_j \not \in A'_{2j}$.
Then, for every $i < j$ we have $g_i \in A'_{2i}$,
hence $\{g_1,\dots,g_j\} \cap A'_{2j} = \emptyset$.
Since, $g_j$ was the most important item in $A_{2j}$,
agent $2j$ prefers it to all bundles that does not have any
good from $\{g_1,\dots,g_j\}$.
Hence, $A_{2j} \succ_{2j} A'_{2j}$, which contradicts~\eqref{eq:thrm:efx:hardness:po}.
Thus, indeed $A'_{2j} \cap G = A_{2j} \cap G$.
%
%
Since in $A'$ all goods stay with agents that had them in $A$,
we know that $A'_i = \emptyset$ for every $i \in [2n]$ such that $A_i = \emptyset$
(otherwise, this would mean that $A'_i$ contains only chores,
which would contradict~\eqref{eq:thrm:efx:hardness:po}).
Finally, observe that every chore $c \in C$
cannot be allocated to agent $2j$
if $c \not \in C^j$.
Otherwise, we would have $A_{2j} \succ_{2j} A'_{2j}$,
which would contradict~\eqref{eq:thrm:efx:hardness:po}.
Hence, $A'$ is an allocation in which
$g_j \in A'_{2j}$, for every $j \in [k]$,
and  $c \in A'_{2j}$ for every $j \in [k]$ and $c \in C^j$.
This means that $A'=A$.
Thus, $A$ is PO.

Finally, let us show that if there is
no exact set cover $\mathcal{K}$ in the original instance,
then there does not exist an \EFX{} allocation.
By contradiction, let us assume otherwise, i.e.,
there is no exact set cover $\mathcal{K}$ in the original instance,
but there is an \EFX{} allocation $A=(A_1,A_2,\dots,A_{2n})$.
We begin by showing that there exists an agent that does not receive any chore in $A$.

\begin{claim}
\label{claim:efx:hardness:nochores}
If $A=(A_1,A_2,\dots,A_{2n})$ is an \EFX{} allocation,
then there exists $j \in [2n]$ such that $A_j \cap C = \emptyset$.
\end{claim}
\begin{proof}
Assume otherwise, i.e., every agent receives at least one chore.
Let us consider an agent that receives good $g_1$ in $A$.
Without loss of generality, let us assume that $g_1 \in A_1$
(otherwise, we can renumerate the subsets and the agents).
Since agents $1$ and $2$ have identical preferences, i.e., $\rhd_1 = \rhd_2$,
we have two cases: either agent $2$ envies agent $1$ (case I),
or agent $1$ envies agent $2$ (case II).

Case I:
Let us denote an arbitrary chore
that agent $2$ receives in $A$ by $c \in A_2$.
If $c \not \in C^1$, then for both agents $1$ and $2$
chore $c$ is more important than good $g_1$.
Hence, since $2$ envies $1$, we have $A_1 \setminus \{g_1\} \succ_2 A_2$.
Thus, $A$ is not \EFX{}---a contradiction.
On the other hand, if $c \in C^1$, then for both $1$ and $2$
chore $c$ is less important than good $g_1$.
Hence, since $2$ envies $1$, even if we remove $c$,
agent $2$ would still envy $1$, i.e., $A_1 \succ_2 A_2 \setminus \{c\}$.
Again, this would mean that $A$ is not \EFX{},
which concludes the analysis of this case.

Case II:
Now, let us assume that agent $1$ envies agent $2$.
Since agent $2$ cannot receive a more important good than $g_1$,
the only case in which $1$ can envy $2$ is if $1$ receives a chore, $c$,
that is the most important item for $1$ and $2$ among all items in $A_1 \cup A_2$.
This implies that $c \not \in C^1$.
Observe that this means that $1$ cannot receive any other chore $c' \in C$.
Otherwise, $1$ would still prefer $A_2$ over $A_1 \setminus \{c'\}$,
which would mean that $A$ is not \EFX{}.
Hence, $A_1$ consists of $c, g_1$ and possibly other goods.
Take arbitrary $j \in [n]$ such that $c \in C^j$
(since every $u_i$ belongs to some $S_j$,
there has to exist one).
Now, observe that agent $2j$, prefers $A_1$ even to the bundle
consisting of all of the remaining goods and no chores, i.e.,
$A_1 \succ_{2j} (G \setminus A_1)$,
because $A_1$ contains $g_1$ and just one chore, $c$,
for which $g_1 \rhd_{2j} c$.
However, by our assumption agent $2j$ has some chore, $c' \in C$.
Therefore, $2j$ envies $1$ even without $c'$,
i.e., $A_1 \succ_{2j} (G \setminus A_1) \succeq_{2j} (A_j \setminus \{c'\})$.
Thus, $A$ is not \EFX{}, which concludes the proof of the claim.
\end{proof}

Next, for every $j \in N$,
by $UC(A_j)$ let us denote the set of \emph{uncovered chores}
received by agent $j$, i.e.,
chores that are more important to $j$
then the most important good they received
(or all $j$'s chores, if $j$ does not have goods).
Formally,
\(
    UC(A_j) = \{c \in A_j \cap C : \{ c \} \rhd_j (G \cap A_j) \}.
\)
Let us show, that no agent can receive more than one uncovered chore.
\begin{claim}
\label{claim:efx:hardness:max1chore}
If $A=(A_1,\dots,A_{2n})$ is an \EFX{} allocation,
then for every $j \in [2n]$,
we have $|UC(A_j)| \le 1$.
\end{claim}
\begin{proof}
Assume by contradiction that $A=(A_1,\dots,A_{2n})$ is an \EFX{} allocation
and there exists $j \in [2n]$ such that $UC(A_j) > 1$.
Let $c \in A_j$ be an arbitrary uncovered chore of agent $j$.
From Claim~\ref{claim:efx:hardness:nochores}
we know that there exists agent $i \in [2n]$
such that $A_i \cap C = \emptyset$.
Hence, since $A_j \setminus \{c\}$
still contains at least one uncovered chore for $j$,
we have that $A_i \succ_j A_j \setminus \{c\}$.
Thus, $A$ is not \EFX{}---a contradiction.
\end{proof}

Now, we will show a contradiction in the total number of uncovered chores
that we can have in allocation $A$.
From Claim~\ref{claim:efx:hardness:max1chore}
we get that $|UC(A_j)| \le 1$,
for every $j \in [2n]$.
Also, from Claim~\ref{claim:efx:hardness:nochores}
we know that there exists $j \in [2n]$
such that $|UC(A_j)|=0$.
Combining both facts together, we obtain an upper bound
\[
    \textstyle \sum_{j \in [2n]} |UC(A_j)| \le 2n - 1.
\]

Next, let us focus on the lower bound.
To this end,
observe that every chore $c \in C$
is not uncovered only if it is
assigned to an agent $j \in [2n]$
that receives a good, i.e., $A_j \cap G \neq \emptyset$,
and $c \in C^{\lfloor j/2 \rfloor}$.
By $K = \{j \in [2n] : A_j \cap G \neq \emptyset\}$ let us denote
the set of agents that received at least one good.
Also, by $\mathcal{K} = \{S_{\lfloor j /2 \rfloor} : j \in K\}$
let us denote the subfamily of subsets associated with agents in $K$.
Since $|\mathcal{K}| \le |K| \le k$ and
there is no set cover of size $k$ in the original instance,
we know that there exists $u_i \in \mathcal{U}$
such that $u_i \not \in S_j$ for every $S_j \in \mathcal{K}$.
Thus, chores $c_{i,1}, \dots , c_{i,2n} \not \in C^{\lfloor j/2 \rfloor}$
for every $j \in K$, which means they are uncovered.
Hence, 
\[
    \textstyle \sum_{j \in [2n]} |UC(A_j)| \ge 2n.
\]
However, this is in contradiction with our upper bound,
which concludes the proof.
\end{proof}

\subsection[Proof of Theorem 2]{Proof of \cref{thrm:efx+po+sep}}
\label{app:efx+po:separable}

\thrmefxposep*

\begin{proof}
 Let $(N, M, G, \rhd)$ be an arbitrary separable instance such that
    $\rhd_i(1) \in C_i$ for every $i \in N$.
    We will show that allocation $A$ returned by
    Algorithm~\ref{alg:efx+po} for this instance
    is \EFX{} and PO and it computes in polynomial time.
    For the latter, observe that the total number of iterations
    in all phases is bounded by $|N| + |M|$.
    Hence, let us focus on proving \EFX{} and PO. 
    We denote the partial allocation obtained at the end of Phase 1 by $B$.
    In what follows, we will first show that $B$ and $A$ are \EFX{},
    and then that $A$ is PO as well.
    Throughout the proof we will use the notation introduced in Algorithm~\ref{alg:efx+po},
    but let us fix $N'$ as the set of agents that did not receive an item in Phase 1
    of the algorithm, i.e., $N' = \{i \in N: B_i = \emptyset\}$.
    

    \underline{\textit{$B$ is \EFX{}}}.
    Fix arbitrary $i, j \in N$ such that $i < j$.
    If $i \in N'$, then $j \in N'$ as well. 
    Thus, $B_i = B_j = \emptyset$ and there is no envy between $i$ and $j$. 
    Hence, assume that $i \notin N'$. 

    If $j \in N'$, then we know that $|\bar{C}| < n$.
    Thus, $B_i$ consists of exactly one common chore $c \in \bar{C}$ and
    goods assigned to it in line 6, i.e., $B_i \setminus \{c\} \subseteq G_i$.
    The instance is separable, hence $c$ is more important for $j$
    than every item it considers as a good i.e., $\{c\} \rhd_j G_j$.
    Thus, $j$ does not envy $i$. 
    Now, $i$ envies $j$.
    However, since $j$ does not have any items and $i$ has only one chore,
    the only item we have to check for \EFX{} is $c$.
    But since $B_i \setminus \{c\} \subseteq G_i$,
    we know that $B_i \setminus \{c\} \succeq_i B_j$.
    Thus, $i$ and $j$ do not violate \EFX{}.
    
    Finally, if $j \notin N'$, then $j$ has exactly
    one common chore $c \in \bar{C}$ in its bundle, i.e.,
    $B_j \setminus \{c\} \subseteq G_j$.
    We know that $i$ will not envy $j$,
    because $i$ picked its chores before $j$ in line 2,
    so $c$ is more important to $i$ than every item in $B_i$
    (otherwise $i$ would pick $c$ instead).
    On the other hand, $j$ might envy $i$.
    However, notice that $j$ received its goods in line 6 before $i$,
    so every item in $B_i$ is a chore for $j$, i.e., $B_i \cap G_j = \emptyset$.
    Thus, the only item we have to check for \EFX{} is chore $c$.
    But since $B_j \setminus \{c\} \subseteq G_j$ and $i$ has some common chore $c'$
    that is more important to $j$ than all goods, i.e., $\{c'\} \rhd_j G_j$,
    we have $B_j \setminus \{c\} \succ_j B_i$.
    Thus, $i$ and $j$ do not violate \EFX{}.

    \underline{\textit{$A$ is \EFX{}}}.
    Fix arbitrary $i, j \in N$ such that $i < j$.
    If $i \in N'$, then $j \in N'$ as well.
    If $A_i = A_j = \emptyset$,
    then $i$ and $j$ do not violate \EFX{}.
    Otherwise, observe that one of them had to
    receive some good in line 10 of the algorithm.
    Let us assume that $i$ received goods in Phase 2
    and it was when $j$ was still ``unhandled'' agent, i.e.,
    $j$'s bundle was empty at that time
    (the case where $j$ received goods before $i$ is analogous).
    Then, $g_i$ is more important for $i$ then any item in $A_j$.
    Thus, $i$ does not envy $j$.
    Now, $j$ may envy $i$.
    However, observe that in Phase 2 we do not assign chores to agents,
    i.e., $A_j \cap C_j = \emptyset$.
    Moreover, the only item that $j$ may consider as good in $A_i$ is $g_i$
    (all other items in $A_i$ are assigned in line 11,
    but we excluded $G_j$ there).
    Hence, the only item we have to check for \EFX{} is $g_i$.
    But since $A_j \subseteq G_j$ and $A_i \cap G_j = \{g_i\}$,
    we have $A_j \succeq_j A_i \setminus \{g_i\}$.
    Hence, $i$ and $j$ do not violate \EFX{}.
    Thus, in the remainder of this part of the proof
    we assume that $i \not \in N'$.
    
    If $j \in N'$, it means that $|\bar{C}| < n$.
    Thus, $A_i$ contains exactly one common chore $c \in \bar{C}$
    and no other chores, i.e., $A_i \setminus \{c\} \subseteq G_i$.
    Instance is separable,
    so $c$ is more important then all goods for $j$.
    Since $A_j$ does not contain any chore, $j$ does not envy $i$.
    Now, $i$ may envy $j$.
    However, recall that $j$ does not receive any item considered as good by $i$,
    i.e., $A_j \cap G_i = \emptyset$.
    Hence, the only item we have to check for \EFX{} is $c$.
    However, since $A_i \setminus \{c\} \subseteq G_i$,
    we have $A_i \setminus \{c\} \succeq_i A_j$.
    Hence, $i$ and $j$ do not violate \EFX{}.

    Finally, if $j \not \in N$,
    then $A_i = B_i$ and $A_j = B_j$.
    Thus, $i$ and $j$ do not violate \EFX{} from the fact that $B$ is \EFX{}.

    \underline{\textit{$A$ is PO}}.
    By $D$ let us denote the partial allocation obtained from $A$
    by removing all common chores i.e.,
    $D_i = A_i \setminus \bar{C}$ for every $i \in N$.
    Observe that in $D$
    all agents have only items they consider as goods,
    i.e,. $D_i \subseteq G_i$ for every $i \in N$.
    Let us first prove that $D$ is PO and then,
    using Lemma~\ref{lmm:po}
    we will show that $A$ is PO as well.
    
    Assume that $D$ is not PO. 
    Then, there exists a partial allocation $D'$ that Pareto dominates $D$.
    First, let us show that $D'_i = D_i$ for every $i \in N \setminus N'$.
    Assume otherwise and take the largest $i \in N \setminus N'$
    such that $D'_i \neq D_i$.
    By Pareto domination, this means that $D'_i \succ_i D_i$.
    Hence, there exists $g \in G_i$
    such that $g \in D'_i \setminus D_i$.
    However, observe that in line 6,
    we assign to $i$ all its goods that were not assigned earlier, i.e.,
    $D_i = G_i \setminus (D_z \cup D_{z-1} \cup \dots D_{i+1})$.
    Since $i$ is the largest,
    $D'_j = D_j$ for every $j \in \{z, z-1, \dots, i+1\}$.
    Hence, there is no other good we can add to $D'_i$---a contradiction.

    Now, let us prove that $D'_i = D_i$ for all $i \in N'$ as well.
    To this end, first by $N^0 = \{i \in N': D_i = \emptyset\}$
    let us denote the set of agents in $N'$ that did not receive a good.
    Now, let us show that $g_i \in D'_i$ for every $i \in N' \setminus N^0$.
    Assume otherwise and among $i \in N' \setminus N^0$ such that $g_i \not \in D'_i$
    let us take the one that received $g_i$ the earliest in the execution of the algorithm.
    Since $D'_i \succ_i D_i$, there must exists $g \in G_i$ such that $g \rhd_i g_i$
    and $g \in D'_i \setminus D_i$.
    Let $j$ be an agent such that $g \in D_j$.
    Since $g \not \in D'_j$ we know that $j \in N'$
    (as $D_k = D'_k$ for all $k \in N \setminus N'$).
    We also know that $j \not \in N^0$.
    Now, if $j$ received $g_j$ after $i$ received $g_i$,
    then $g_i$ is more important to $i$ than all items in $D_j$
    (otherwise $i$ would get one of them).
    But this is a contradiction since $g \in D_j$
    and we assumed $g \rhd_i g_i$.
    On the other hand, if $j$ received $g_j$ before $i$ received $g_i$,
    then from the assumption that $i$ is the earliest,
    we know that $g_j \in D'_j$.
    However, all other goods that $j$ received in line 11,
    are not viewed as goods by $i$, in particular $g \not \in G_i$.
    But this is again a contradiction.
    Thus, we have proven that $g_i \in D'_i$ for every $i \in N' \setminus N^0$.
    
    Now, let us show that in fact $D'_i = D_i$ for every $i \in N' \setminus N^0$.
    Assume otherwise and among $i \in N' \setminus N^0$ such that $D'_i \neq D_i$
    let us take such that $i$ was given $g_i$ as the latest.
    Since $D'_i \succ_i D_i$, there exists $g \in G_i$
    such that $g \in D'_i \not \in D_i$.
    Let $j$ be such that $g \in D_j$.
    Since $i$ was given $g_i$ as the latest
    and $D'_j \neq D_j$,
    we know that $j$ was given $g_j$
    before $i$ received $g_i$.
    We also know that $g \neq g_j$ since $g_j \in D'_j$.
    However, all other items given to $j$ in line 11,
    are not viewed as goods by $i$.
    In particular $g \not \in G_i$,
    which is a contradiction.

    Since for every agent $i \in N \setminus N^0$,
    we know that $D'_i = D_i$,
    this means that the agents that did not receive items in $D$
    do not receive them in $D'$, i.e., $D'_i = \emptyset = D_i$ for each $i \in N^0$.
    Thus, $D' = D$, which contradicts the assumption that $D'$ Pareto dominates $D$.
    Therefore, $D$ is PO.
    
    Finally, observe that since $D_i \subseteq G_i$ for every $i \in N$
    and we have a separable instance in which the most important item of every agent is a chore,
    this means that $\bar{C} \rhd_i D_i$.
    Hence, $A$ is PO from Lemma~\ref{lmm:po}.
\end{proof}

\section[Omitted Material from Section 4]{Omitted Material from \cref{sec:EF1+PO}} \label{app:EF1+PO}

\subsection[Examples for the EF1 and PO Algorithm]{Examples for the EF1 and \PO{} Algorithm}
\label{app:ef1+po:examples}

In this section,
we analyze several examples
that show the complexity of the problem
of finding an \EF{1} and \PO{} allocation.
Specifically,
\begin{itemize}
    \item in \cref{ex5:ef1+po},
    we show that a double round robin algorithm
    as proposed by \cite{ACI+19fair}
    does not necessarily find a \PO{} allocation,
    \item in \cref{ex6:ef1+po},
    we argue that a modified version of the envy-graph algorithm
    as proposed by \cite{BSV21approximate}
    does not always find a \PO{} allocation as well,
    \item in \cref{ex4:ef1+po},
    we demonstrate that
    \cref{alg:ef1+po:n-1} that finds \EF{1} and \PO{}
    allocation when there are at least
    $(n-1)$ common terrible chores
    may not return an \EF{1} allocation,
    when this condition does not hold,
    \item finally, in \cref{ex1:ef1+po},
    we show that bundling a good with all less important chores
    (which is a building block of
    an algorithm proposed by \cite{hosseini2022fairly})
    may preclude achieving an \EF{1} allocation.
\end{itemize}

\begin{example}
\label{ex5:ef1+po} 
Consider a mixed instance with two agents, five items and a profile as follows. The set of common goods is $\bar{G} = \{o_{2}^{+}, o^+_3\}$, and the set of common chores is $\bar{C} = \{o_{1}^{-}, o_{4}^{-}, o_{5}^{-}\}$. Notice that in this instance, there is one ($n-1$) common terrible chore, which is $o^-_1$. 

    \begin{align*}
       1 &: \quad o^-_1 \rhd o^+_2 \rhd \underline{o^+_3} \rhd \underline{o^-_4} \rhd o^-_5 \\
       2 &: \quad \underline{o^-_1} \rhd o^-_4 \rhd \underline{o^-_5} \rhd \underline{o^+_2} \rhd o^+_3 
    \end{align*}
This example shows that even though the double round robin algorithm \citep{ACI+19fair} efficiently computes an \EF{1} allocation under additive valuations, it does not necessarily return a \PO{} allocation
(the allocation returned by this algorithm is underlined). 

The algorithm starts by allocating common chores $(\bar{C})$ according to a round-robin sequence $(1,2)$. Since there are three common chores to be allocated among two agents, the algorithm adds one dummy null item (i.e., an item with zero valuation for both agents) to $\bar{C}$, and then agents come in a round-robin sequence $(1,2)$ and pick their most preferred item left in $\bar{C}$. Thus, agent $1$ receives $\{o^-_4\}$ while agent $2$ receives $\{o^-_1, o^-_5\}$. Then, in the reverse round-robin sequence $(2,1)$, agents pick their most preferred item left in $\bar{G}$. Therefore, agents $2$ and $1$ get $o^+_2$ and $o^+_3$, respectively.

This allocation is not \PO{}, as taking $\{o^+_2, o^-_5\}$ from agent $2$ and giving it to agent $1$ will lead to a Pareto improvement. 


\end{example}

\begin{example}
\label{ex6:ef1+po} 
We revisit the mixed instance given in Example~\ref{ex5:ef1+po}

    \begin{align*}
       1 &: \quad \underline{o^-_1} \rhd \underline{o^+_2} \rhd o^+_3 \rhd o^-_4 \rhd \underline{o^-_5} \\
       2 &: \quad o^-_1 \rhd \underline{o^-_4} \rhd o^-_5 \rhd o^+_2 \rhd \underline{o^+_3} 
    \end{align*}
This example also shows that even though the modified version of the envy-graph algorithm \citep{BSV21approximate} efficiently computes an \EF{1} allocation under doubly monotone valuations, it does not necessarily return a \PO{} allocation
(the allocation returned by this algorithm is underlined).

The algorithm starts by allocating common goods $(\bar{G})$ using the envy-cycle elimination algorithm of \citet{lipton2004approximately} in Phase 1. Thus, agents $1$ and $2$ receive $o^+_2$ and $o^+_3$, respectively. Then, in Phase 2, as long as there is an agent who has no outgoing edge (i.e., a sink) in the envy-graph, that agent receives a common chore from $\bar{C}$. Therefore, agent $1$ receives $\{o^-_5, o^-_1\}$ and agent $2$ receives $o^-_4$.

This allocation is not \PO{}, as taking $\{o^+_3, o^-_4\}$ from agent $2$ and giving it to agent $1$ will lead to a Pareto improvement. 

\end{example}

\begin{example}
\label{ex4:ef1+po}
 Consider a mixed instance with three agents, four items and a profile as follows. The set of common goods is $\bar{G} = \{o_{1}^{+}\}$, and the set of common chores is $\bar{C} = \{o_{2}^{-}, o_{3}^{-}, o_{4}^{-}\}$. Notice that in this instance, there is one ($n-2$) common terrible chore, which is $o^-_2$. 
 \begin{align*}
       1 &: \quad \underline{o^-_2} \rhd \underline{o^+_1} \rhd o^-_3 \rhd o^-_4 \\
       2 &: \quad o^-_2 \rhd o^+_1 \rhd o^-_3 \rhd o^-_4  \\ 
       3 &: \quad o^-_2 \rhd o^+_1 \rhd \underline{o^-_3} \rhd \underline{o^-_4} 
    \end{align*}
This example shows that Algorithm~\ref{alg:ef1+po:n-1} does not necessarily return an \EF{1} allocation when the number of common terrible chores is less than $n-1$
(the allocation returned by this algorithm is underlined).

Suppose the ordering is $(1,2,3)$.
\cref{alg:ef1+po:n-1} starts by assigning $o^+_1$ to agent $1$.
Then, in the reverse ordering $(3,2,1)$, agents get their common non-terrible chores (out of the remaining items), resulting in agent $3$ receiving $\{o^-_3, o^-_4\}$ (and nothing for others).
Since agent $3$ is not envied (such an agent always exists), \cref{alg:ef1+po:n-1} allocates a common terrible chore ($o^-_2$) by running a serial dictatorship with the ordering of $(1,2,3)$ and single quota.

This allocation is not \EF{1}, as agent $3$ whose bundle consists of two chores envies agent $2$ who has an empty bundle
and this envy cannot be eliminated by removal of one item.
\end{example}

\begin{example}
\label{ex1:ef1+po}
 Consider a mixed instance with four agents, seven items and a profile as follows. The set of common goods is $\bar{G} = \{o_{1}^{+}\}$, and the set of common chores is $\bar{C} = \{o_{2}^{-}, o_{3}^{-}, o_{4}^{-}, o_{5}^{-}, o_{6}^{-}, o_{7}^{-}\}$
 (note that agents 1 and 2 have identical preferences
 and the same holds for agents 3 and 4).

    \begin{align*}
       1 &: \quad o^-_2 \rhd o^-_3 \rhd o^-_4 \rhd
        \underline{o^+_1} \rhd \underline{o^-_5} \rhd \underline{o^-_6} \rhd \underline{o^-_7} \\
       2 &: \quad o^-_2 \rhd o^-_3 \rhd o^-_4 \rhd o^+_1 \rhd o^-_5 \rhd o^-_6 \rhd o^-_7 \\ 
       3 &: \quad o^-_5 \rhd o^-_6 \rhd o^-_7 \rhd o^+_1 \rhd \underline{o^-_2} \rhd o^-_3 \rhd o^-_4 \\
       4 &: \quad o^-_5 \rhd o^-_6 \rhd o^-_7 \rhd o^+_1 \rhd o^-_2 \rhd \underline{o^-_3} \rhd \underline{o^-_4}  
    \end{align*}
We will show that giving an agent a good, $o^+_1$,
together with all of the chores
that are less important then $o^+_1$ for this agent,
necessarily results in an allocation that is not \EF{1}
(an example of such allocation is underlined).

Observe that the instance is symmetric,
hence without loss of generality
we can assume that we give good $o^+_1$ to agent 1.
Now, an intuitive idea for guarantying fairness
is to give to agent 1 also all of the chores
that are less important than $o^+_1$
in order to balance its bundle.
In fact, such technique is used by \cite{hosseini2022fairly}
in the algorithm for finding an \EFX{} and \PO{} allocation
in instances where there is an agent without terrible chores.

However, we can show that assignment of $o^+_1, o^-_5, o^-_6$, and $o^-_7$
to agent 1 cannot be extended to an \EF{1} allocation.
To see this, first observe that in such a case
assigning any chore $o^-_2, o^-_3,$ or $o^-_4$
to agent 2 would violate \EF{1}.
But if we assign all three chores to agents 3 and 4,
then one of them will receive two chores.
Then, this agent will envy agent 2
even after removal of one of the chores.
Hence, indeed, an \EF{1} allocation is unobtainable
in this way.

Observe that to get an \EF{1} and \PO{} allocation
it suffices to move $o^-_7$ from agent 1 to agent 2
in the underlined allocation.

\end{example}





\subsection[Proof of Theorem 3]{Proof of \cref{thm:ef1POterrible}}
\label{app:ef1+po:terrible}

In this section,
we provide the pseudocode for \cref{alg:ef1+po:n-1}
and prove the correctness of the algorithm.

\thmefonePOterrible*

\begin{proof}
    Let $(N, M, G, \rhd)$ be an arbitrary lexicographic mixed instance
    with at least $n-1$ common terrible chores,
    i.e., $\bar{C}^* \ge n-1$.
    We assume that $N=[n]$,
    but this is without loss of generality
    as we can relabel the agents.
    We will show that allocation $A$ returned
    by Algorithm~\ref{alg:ef1+po:n-1} for this instance
    is \EF{1} and \PO{} and it computes in polynomial time.
    For the latter, observe that
    the total number of iterations in Phase 1 is bounded by $2|N|$
    and in Phase 2 by $|M|$.
    Hence, let us focus on showing that $A$ is \PO{} and \EF{1}.
    By $B$ let us denote the partial allocation we obtain after the end of Phase 1.
    Since $B$ is \PO{} (as we will show),
    there always exists an agent that is not envied by any other agent
    (otherwise there has to be a cycle of envy,
    which resolved would be a Pareto improvement).
    Thus, there exists ordering $\sigma$
    such that no one envies agent $\sigma_n$
    (let us take the one that the algorithm takes in line 9).
    Finally, by $H$ let us denote all of the common chores
    that are not terrible for some agent, i.e., $H = \bar{C} \setminus \bar{C}^*$.
    In what follows we first show that allocation $B$ is \PO{}
    and then that allocation $A$ is \PO{} and \EF{1}.

    \underline{\textit{$B$ is \PO{}}}.
    Assume that $B$ is not \PO{},
    i.e., there exists an allocation $B'$
    that Pareto dominates $B$.
    Let $i$ be the agent with minimal number that
    receives different bundles in $B$ and $B'$.
    By Pareto domination, this means that $B'_i \succ_i B_i$.
    Hence, there exists a good in $B'_i$ that is not in $B_i$
    or chore in $B_i$ that is not in $B'_i$.

    If there exists a good $g \in G_i$ such that
    $g \in B'_i \setminus B_i$,
    then let $j$ be the agent that has $g$ in $B$.
    This means that $B'_j \neq B_j$,
    which by $i$'s minimality implies that $j > i$.
    This means that $g$ was not given to any agent
    with the number smaller than $i$,
    i.e., $g \not \in A_k$ for every $k \in [i-1]$.
    But since $g \in G_i$,
    good $g$ should be then assigned to agent $i$
    in line 3 before $j$'s turn---a contradiction.

    If there exists a chore $c \in C_i$ such that
    $c \in B_i \setminus B'_i$,
    then observe it is a common chore
    (only common chores are assigned as chores in the algorithm).
    Let $j$ be an agent that receives $c$ in $B'$.
    This means that $B'_j \neq B_j$,
    which by $i$'s minimality implies that $j > i$.
    Note that since $c$ was not assigned to $j$
    in line 6 of the algorithm
    (and $j$ was handled there before $i$),
    it means that $c$ is a terrible chore for $j$,
    i.e., $c \in C^*_j$.
    But then, $c$ is more important to all goods for $j$
    and thus also all items it received in Phase 1,
    i.e., $\{c\} \rhd_j G_j \cup B_j$.
    This implies that $B_j \succ_j B'_j$,
    which contradicts the Pareto domination.
    
    \underline{\textit{$A$ is \PO{}}}.
    Observe that $\bar{C}^*$ is more important to every agent
    than all items it considers as goods.
    Moreover, in Phase 1, we assign
    only goods and chores less important than goods.
    Hence, $\bar{C}^* \rhd_i G_i \cup B_i$ for every $i \in N$.
    Since $B$ is \PO{}, from Lemma~\ref{lmm:po}
    we know that $A$ is \PO{} as well.

    \underline{\textit{$A$ is \EF{1}}}.
    Let us take $i, j \in N$ such that $i$ is before $j$
    in ordering $\sigma$.
    
    If $j \neq \sigma_n$ or $|\bar{C}^*|>n-1$,
    then agent $j$ receives exactly one chore, $c \in \bar{C}^*$, in Phase 2
    of the algorithm.
    Agent $i$ does not envy $j$, because $c$ is more important chore
    than all items in $i$'s bundle
    (otherwise, since $i$ picks in serial dictatorship before $j$,
    $i$ would choose $c$ instead of its chore).
    Now, agent $j$ can envy agent $i$.
    However, agent $i$ has some common terrible chore $c' \in \bar{C}^*$.
    This chore is more important to $j$,
    than all goods and anything it received in Phase 1,
    i.e., $\{c'\} \rhd_j B_j \cup G_j$.
    Hence, $A_j \setminus \{c\} = B_j \succ_j A_i$,
    which means that $i$ and $j$ do not violate \EF{1}.

    If $j = \sigma_n$ and $|\bar{C}^*|=n-1$,
    agent $j$ does not receive any chore in Phase 2,
    i.e., $A_j = B_j$,
    while agent $i$ receives exactly one chore, $c \in \bar{C}^*$.
    Observe that $c$ is more important to $j$ than
    all goods and anything it received in Phase 1,
    i.e., $\{c\} \rhd_j B_j \cup G_j$.
    Hence, agent $j$ does not envy agent $i$.
    Now, agent $i$ envies agent $j$.
    However, recall that we have chosen $\sigma$ in such a way
    that no agent envies $\sigma_n$ in partial allocation $B$.
    Therefore, $A_i \setminus \{c\} = B_i \succeq_i B_j = A_j$,
    which means that $i$ and $j$ do not violate \EF{1}.
\end{proof}

\section[Omitted Material from Section 5]{Omitted Material from \cref{sec:mms+po}}

\subsection[Proof of Theorem 4]{Proof of \cref{thrm:mms+po}}
\label{app:mms+po}

\thrmmmspo*

\begin{proof}
    Let $(N, M, G, \rhd)$ be an arbitrary lexicographic mixed instance with terrible chores,
    i.e., $\rhd_i(1) \in C_i$ for each $i \in N$.
    We assume that $N=[n]$,
    but this is without loss of generality
    as we can relabel the agents.
    We will show that allocation $A$ returned
    by Algorithm~\ref{alg:mms+po} for this instance
    is \MMS{} and PO and it computes in polynomial time.
    For the latter, observe that
    the total number of iterations is bounded by $|N|$.
    Hence, let us focus on first showing \MMS{} and then PO.
    Throughout the proof we will use the notation introduced in the algorithm.

    \underline{\textit{MMS}}.
    Since the most important item for every agent is a chore,
    by Proposition~\ref{prop:mms_threshold},
    we have that $\MMS{}_i = \rhd_i(1) \cup G_i$
    for every agent $i \in N$.
    Thus, for every agent $i \in N$ that does not receive
    its most important item, i.e., $\rhd_i(1) \not \in A_i$,
    we have $A_i \succ_i MMS_i$.
    Moreover, the only agent
    that can receive its most important item 
    is agent 1.
    However, we give agent 1 all of the items it considers as goods in line 3
    and we do not give it any chores,
    except for possibly $c^*$ in line 7.
    Hence, $A_1 \succeq_1 MMS_1$,
    which means that $A$ satisfies \MMS{}.


    \underline{\textit{PO}}.
    Assume that $A$ is not PO, i.e.,
    there exists allocation $A'$
    that Pareto dominates $A$.
    By $i^*$ let us denote the agent that receives
    item $c^*$ in allocation $A$.
    Also, let us take the smallest $i \in [n]$
    such that agent $i$ receives different bundles
    in allocations $A$ and $A'$.
    From PO domination,
    this means that $A'_i \succ_i A_i$.

    If $i < i^*$, then we know that the only items
    received by $i$ in the algorithm
    are the goods assigned to it in line 3, i.e.,
    $A_i = G_i \setminus \bigcup_{j \in [i-1]} A_j$.
    Since $i$ is the smallest,
    we have $A'_j = A_j$ for every $j \in [i-1]$.
    Thus, $A_i$ is the best possible bundle agent $i$ can receive.
    Hence, $A'_i \not\succ_i A_i$---a contradiction.

    If $i = i^*$ and $c^*$ is a good for agent $i$,
    then we can get the contradiction
    in exactly the same way as for $i < i^*$.
    Hence, assume that $i = i^*$ and
    $c^*$ is a chore for $i$.
    We will first show that $c^* \in A'_i$.
    Assume otherwise, i.e., there exists agent $j \in N \setminus \{i\}$
    such that $c^* \in A'_j$.
    Since $c^* \not \in A_j$ and
    $i$ is the agent with the smallest number
    that receives different bundles in $A$ and $A'$,
    we know that $j > i$.
    On the other hand,
    since $c^*$ is a chore for agent $i$,
    $c^*$ has to be the most important item for all $j > i$,
    i.e., $\rhd_j(1) = c^*$.
    But that would mean that $A_j \succ_j A'_j$,
    which is a contradiction.
    Now, let us show that the remaining items in $A_i$ and $A'_i$
    must also be the same.
    Observe that apart from $c^*$,
    the only items received by $i$ are goods that it gets in line 3.
    Hence, $A_i = \{c^*\} \cup G_i \setminus \bigcup_{j \in [i-1]} A_j$.
    But since $c^* \in A'_i$ and $A'_j = A_j$ for every $j \in [i-1]$
    this is the best possible bundle agent $i$ can receive.
    Thus, $A'_i \not\succ_i A_i$---a contradiction.
    
    If $i^* < i < n$,
    we can get the contradiction
    in exactly the same way as for $i < i^*$.

    Finally, observe that $i$ cannot be equal to $n$,
    as there always have to be at least two agents
    with different bundles in two different allocations.
\end{proof}

\subsection[Proof of Theorem 5]{Proof of \cref{thrm:mms+rm}}
\label{app:mms+rm}

Let us begin by introducing formal definition of rank maximallity (RM)
and required notation.

By a \emph{position} of item $o \in M$ in agent's $i$ importance ordering, we mean a number $\pos_i(o)$ such that $\rhd_i(\pos_i(o)) = o$.
Allocation $A$ is \emph{rank maximal} (RM), 
    if for every common chore $c \in \bar{C}$
    it holds that $c \in A_i$ for some agent $i \in N$ such that
    $\pos_i(c) = \max_{j\in N} \pos_j(c)$,
    and for every other item $o \in M \setminus \bar{C}$,
    it holds that $o \in A_i$ for some agent $i \in N^o$,
    where $N^o = \{j \in N: o \in G_j\}$,
    and $\pos_i(o) = \min_{j \in N^o} \pos_j(o)$. 

\begin{example}
We revisit the mixed instance given in Example~\ref{ex5:ef1+po} to show a
rank maximal allocation. 
    \begin{align*}
       1 &: \quad o^-_1 \rhd \underline{o^+_2} \rhd \underline{o^+_3} \rhd \underline{o^-_4} \rhd \underline{o^-_5} \\
       2 &: \quad \underline{o^-_1} \rhd o^-_4 \rhd o^-_5 \rhd o^+_2 \rhd o^+_3 
    \end{align*}
The underlined allocation is \RM{}
since every common chore, $o^-_1, o^-_4$ and $o^-_5$,
is allocated to an agent that has it in the position with the highest number
and every common good, $o^+_2$ and $o^+_3$,
is allocated to an agent that has it in the position with the lowest number.
Observe that the only other \RM{} allocation
is the one in which agent 1 receives all of the items.
\end{example}

\thrmmmsrm*

\begin{proof}
Checking whether a given allocation is \MMS{} and RM
can be done in polynomial time,
hence our problem is in NP.
Thus, let us focus on showing the hardness of the problem.

To this end, we will follow a reduction from \textsc{Set Cover}.
In this problem,
we are given a constant $k \in \mathbb{N}$,
a universe of elements,
$\mathcal{U} = \{u_1,\dots,u_m\}$,
and a family of subsets of the universe,
$\mathcal{S} = \{S_1,\dots,S_n\}$.
We assume that every element $u_1,\dots,u_m$
belongs to at least one subset from $\mathcal{S}$.
The goal is to decide whether there exists
a subfamily of the size at most $k$
that covers every element of the universe, i.e.,
$\mathcal{K} \subseteq \mathcal{S}$
such that $|\mathcal{K}| \le k$ and
\(
    \bigcup_{S_j \in \mathcal{K}} = \mathcal{U}.
\)
Such task is known to be NP-complete.

For every instance of this problem,
we define the corresponding
lexicographic mixed instance as follows.
Let us take one agent, $j$, for every set $S_j \in \mathcal{S}$,
$m$ dummy agents, $n+1,\dots,n+m$,
and two additional agents,
called \emph{filler} agent $f$
and \emph{setting} agent $x$.
Hence, we have $N = \{x,f,1,2,\dots,n + m\}$.
Next, let us define the items in our instance---there will be four types of such.
First, let us have $m$ common chores $c_1,\dots,c_m$
corresponding to the elements of the universe $\mathcal{U}$.
Next, let us have $m$ \emph{filler} common goods $f_1,\dots,f_m$ that
correspond to the elements of the universe $\mathcal{U}$ as well
(but their assignment will be straightforward).
Moreover, let us have $k$ common goods $g_1,\dots,g_k$ the
assignment of which will correspond
to the choice of $k$ sets from the family $\mathcal{S}$.
Finally, let us add one common chore $c_x$ and one common good $g_x$
with which we will be able to restrict the number of possible
MMS and RM allocations.
This gives us $2m + k + 2$ items in total
and every one of them is either a common good or a common chore,
i.e., for every $j \in N$,
$G_j = \{f_1,\dots,f_m,g_x,g_1,\dots,g_k\}$ and
$C_j = \{c_x, c_1,\dots,c_m\}$.

Now, let us define the importance orderings of the agents.
First, for brevity, we define the following shorthand notation:
for every linear order $Z = z_1 \rhd z_2 \rhd \dots \rhd z_r$,
and two different elements of this order $z_i$, $z_j$ such that $z_i \rhd z_j$,
by $s(z_i,z_j,Z)$, we will understand a linear order obtained from $Z$
by swapping the positions of the elements $z_i$ and $z_j$ in that order, i.e.,
\(
    s(z_i,z_j,Z) = 
    z_1 \rhd z_2 \rhd \dots
        \rhd z_{i-1} \rhd z_j \rhd z_{i+1} \rhd \dots
        \rhd z_{j-1} \rhd z_i \rhd z_{j+1} \rhd \dots
        \rhd z_r.
\)

Let us begin by defining the importance ordering of the setter agent as follows:
\[
    x: \ \ g_x \rhd c_1 \rhd \dots \rhd c_m \rhd c_x \rhd f_1 \rhd \dots \rhd f_m \rhd g_1 \rhd \dots \rhd g_k. 
\]
Next, let us define the importance ordering of the filler agent as:
\[
    f: \ \ f_1 \rhd \dots \rhd f_m \rhd c_x \rhd g_x \rhd c_1 \rhd \dots \rhd c_m \rhd g_1 \rhd \dots \rhd g_k. 
\]
For the remaining agents, let us first define the standard importance ordering,
which we will then modify for each of the agents independently.
The standard importance ordering is given by:
\[
    T = g_x \rhd c_x \rhd c_1 \rhd \dots \rhd c_m \rhd g_1 \rhd \dots \rhd g_k \rhd f_1 \rhd \dots \rhd f_m. 
\]
Now, for every $j \in [n]$,
let us denote the elements of subset $S_j$ by
$S_j = \{u_{j,1},u_{j,2},\dots,u_{j,m_j}\}$.
Then, we define the importance ordering of agent $j$ as:
\[
    j : \quad s(c_{j,1},f_{j,1},s(c_{j,2},f_{j,2},s(\dots,s(c_{j,m_j},f_{j,m_j},T)\dots))),
\]
i.e., the standard importance ordering $T$
with the positions of $c_i$ and $f_i$ swapped
for every $u_i \in S_j$.
Finally, we set the importance ordering
of every dummy agent $j \in \{n+1,\dots,n+m\}$
to standard ordering $T$, i.e.,
\[
    j: \ \ g_x \rhd c_x \rhd c_1 \rhd \dots \rhd c_m \rhd g_1 \rhd \dots \rhd g_k \rhd f_1 \rhd \dots \rhd f_m. 
\]

In what follows,
we will prove that there exists an allocation that satisfies \MMS{} and RM,
if and only if,
there exists a set cover in the original instance.
To this end,
we first characterize \MMS{} allocations (Claim~\ref{claim:mms+rm:mms}) and RM allocations (Claim~\ref{claim:mms+rm:rm}) .

\begin{claim}
\label{claim:mms+rm:mms}
An allocation $A = (A_x, A_d, A_1, \dots, A_{n+m})$ is \MMS{},
if and only if, for every agent,
either the most important item in its bundle is a good,
or its bundle is empty.
\end{claim}
\begin{proof}
Every item is either a common good or a common chore,
hence for every agent $j \in N$, we have $|G_j| = k+m$.
$k \le n$ and we have $n + m + 2$ agents,
thus for every agent there is less goods then agents.
Since the importance order of every agent starts with a good,
by Proposition~\ref{prop:mms_threshold}
we have that the maximin share allocation for every agent is an empty allocation.
Thus, in an allocation that satisfies maximin share,
every agent has to receive either its maximin share bundle,
which is an empty set,
or a better bundle,
which is any bundle in which the most important item is a good.
\end{proof}

\begin{claim}
\label{claim:mms+rm:rm}
An allocation $A = (A_x, A_d, A_1, \dots, A_{n+m})$ is RM,
if and only if,
\begin{enumerate}[label = \alph*)]
    \item the setting agent receives chore $c_x$, i.e, $c_x \in A_x$,
    \item good $g_x$ is given to any agent apart from the filler agent, i.e., $g_x \not \in A_f$,
    \item the filler agent receives all filler goods, i.e., $f_1,\dots,f_m \in A_f$,
    \item goods $g_1,\dots,g_k$ are received by agents $[n+m]$, i.e.,
        $g_i \in \bigcup_{j = 1}^{n+m} A_j$, for every $i \in [k]$, and
    \item for each $i \in [m]$, chore $c_i$ is received by an agent $j \in [n]$
        such that $u_i \in S_j$.
\end{enumerate}
\end{claim}
\begin{proof}
Every item is either a common good or a common chore,
hence allocation $A$ is rank maximal if and only if
every common good is given to an agent for which it is in the highest position,
and every common chore is given to an agent for which it is in the lowest position.
We  first prove that rank maximality implies points \textit{a}--\textit{e}.

For \textit{a},
observe that $c_x$ is ranked in position $m+2$ by the setting agent $x$,
in position $m+1$ by the filler agent $f$,
and in position 2, by every other agent.
Thus, in a rank maximal allocation, $c_x$ has to be given to agent $x$.

For \textit{b},
note that $g_x$ is ranked in position $m+2$ by the filler agent $f$
and in position $1$ by all other agents.
Hence, in a rank maximal allocation, $g_x$ can be given to every agent except $f$.

For \textit{c},
fix $i \in [m]$
and observe that
\[ 
    \pos_j(f_i) = \begin{cases}
      m + 2 + i,         & \mbox{if } j = x,\\
      i,                 & \mbox{if } j = f,\\
      2 + i,             & \mbox{if } j \in [n] \mbox{ and } u_i \in S_j,\\
      m + k + 2 + i,     & \mbox{otherwise.}
    \end{cases}
\]
Thus, the smallest position
is in the importance ordering of the filler agent $f$.
Hence, in a rank maximal allocation, $f_i$ has be given to agent $f$,
for every $i \in [m]$.

For \textit{d},
fix $i \in [k]$
and observe that
\[ 
    \pos_j(g_i) = \begin{cases}
      2m + 2 + i,        & \mbox{if } j = x,\\
      2m + 2 + i,        & \mbox{if } j = f,\\
      m + 2 + i,         & \mbox{otherwise.}
    \end{cases}
\]
Hence, the smallest position
is in the importance ordering of the agents $1,\dots,n+m$.
Therefore, in an \RM{} allocation,
every good $g_1,\dots,g_k$
has to be given to one of these agents.

Finally, for \textit{e},
fix $i \in [m]$
and observe that
\[ 
    \pos_j(c_i) = \begin{cases}
      1 + i,            & \mbox{if } j = x,\\
      m + 2 + i,        & \mbox{if } j = f,\\
      m + k + 2 + i,    & \mbox{if } j \in [n] \mbox{ and } u_i \in S_j,\\
      2 + i,            & \mbox{otherwise.}
    \end{cases}
\]
Hence, the highest position number
is in the importance ordering of an agent $j \in [n]$
such that $u_i \in S_j$.
Since we assume that for every $u_i \in \mathcal{U}$
there exists $S_j \in \mathcal{S}$ such that $u_i \in S_j$,
we know that such agent exists.
Therefore, in a rank maximal allocation,
every chore $c_1,\dots,c_m$
has to be given to such an agent.

Points \textit{a}--\textit{e} cover all items in the instance.
Hence, if every item is assigned in accordance with them,
then the allocation is rank maximal.
\end{proof}

Now, let us show that if there exists a set cover $\mathcal{K}$ in the original instance,
then there exists an \MMS{} and RM allocation in the corresponding lexicographic mixed instance.
Without loss of generality, we can assume that $\mathcal{K} = \{S_{1}, \dots, S_{k}\}$
(otherwise we can reorder the subsets and corresponding agents).
Then, for every subset $S_j \in \mathcal{S}$,
by $F(S_j)$ let us denote the subset of elements of $S_j$
for which $S_j$ has the smallest index
among all subsets they belong to, i.e.,
$F(S_j) = \{u_i \in S_j: \forall_{j' < j} u_i \not \in S_{j'}\}$.
Next, let us consider allocation $A=(A_x,A_f,A_1,\dots,A_{n+m})$, given by
\[
    A_j =
    \begin{cases}
        \{g_x, c_x\},                        & \mbox{if } j = x,\\
        \{f_1, \dots, f_m\},                 & \mbox{if } j = f,\\
        \{g_j\} \cup \{c_i : u_i \in F(S_j)\},  & \mbox{if } j \in \{1, \dots, k\},\\
        \emptyset,                           & \mbox{otherwise.}
    \end{cases}
\]
Observe that since $\mathcal{K}$ is the cover of $\mathcal{U}$,
every element $u_i$ belongs to exactly one set from $F(S_1),\dots,F(S_k)$.
Thus, $A$ is indeed a well-defined allocation.
Finally, observe that from Claims~\ref{claim:mms+rm:mms} and \ref{claim:mms+rm:rm}
allocation $A$ is both \MMS{} and RM.

It remains to show that if
there exists an \MMS{} and RM allocation $A$ in
the lexicographic mixed instance,
then there exists a set cover in the original instance.
Let us fix such an allocation $A=(A_x,A_f,A_1,\dots,A_{n+m})$.
From Claim~\ref{claim:mms+rm:rm}a we know that $c_x \in A_x$.
By Claim~\ref{claim:mms+rm:mms} this implies
that $A_x$ contains a good
which is more important than chore $c_x$ for the setter agent $x$.
However, the only such good is $g_x$.
Hence, $g_x \in A_x$.
Next, let us denote the agents
which received any of the goods $\{g_1,\dots,g_k\}$
by $\mathcal{H}$, i.e.,
let $\mathcal{H} = \{ j \in N : \{g_1,\dots,g_k\} \cap A_j \neq \emptyset\}$.
From Claim~\ref{claim:mms+rm:rm}d we know that $\mathcal{H} \subseteq [n+m]$.

Now, let us prove
that only agents from $\mathcal{H}$
can receive chores $c_1,\dots,c_m$ in allocation $A$.
Assume otherwise, i.e., there exists agent $j$ and chore $c_i$ such that
$j \not \in \mathcal{H}$ and $c_i \in A_j$.
From Claim~\ref{claim:mms+rm:mms} this means that $A_j$ contains also some good
that is more important for $j$ than chore $c_i$.
However, this good cannot be $g_x$, because as we established $g_x \in A_x$.
Moreover, by Claim~\ref{claim:mms+rm:rm}c, this cannot be one of the goods $f_1,\dots,f_m$.
The only remaining goods are $g_1,\dots,g_k$---a contradiction.
Thus, indeed, all chores $c_1,\dots,c_m$ are distributed among agents from $\mathcal{H}$.
By Claim~\ref{claim:mms+rm:rm}e, this means that
chores $c_1,\dots,c_m$ are given to agents from $\mathcal{H}' = \mathcal{H} \cap \{1,\dots,n\}$.
Since there are $k$ goods $g_1,\dots,g_k$, we know that $|\mathcal{H}'| \le k$.
Therefore, let us take $\mathcal{K} = \{S_j \in \mathcal{S} : j \in \mathcal{H}'\}$.
Clearly, $|\mathcal{K}| \le k$.
Thus, it remains to check that the union of subsets in $\mathcal{K}$
contains all elements of $\mathcal{U}$.
Fix arbitrary $u_i \in \mathcal{U}$.
By Claim~\ref{claim:mms+rm:rm}e, $c_i \in A_j$ for some agent $j$ such that $u_i \in S_j$.
By Claim~\ref{claim:mms+rm:mms}, agent $j$ has to receive some good 
and since it cannot be $g_x$ nor any of $f_1,\dots,f_m$
it has to be one from $g_1,\dots,g_k$.
Hence $j \in \mathcal{H}'$ and $u_i \in \bigcup_{S_j \in \mathcal{K}} S_j$,
which concludes the proof.
\end{proof}

\end{document}